\documentclass{vldb}
\usepackage{color}

\usepackage{times}
\usepackage{algorithmic}
\usepackage{algorithm}

\newcommand{\eat}[1]{}
\newcommand{\ie}{{\em i.e.}}
\newcommand{\eg}{{\em e.g.}}

\newcommand{\tightlist}{\itemsep=-3pt}

\newcommand{\rbox}{\hfill $\Box$}

\newtheorem{definition}{Definition}[section]
\newtheorem{proposition}[definition]{Proposition}
\newtheorem{lemma}[definition]{Lemma}

\newtheorem{theorem}[definition]{Theorem}
\newtheorem{example}[definition]{Example}
\newtheorem{property}[definition]{Property}

\begin{document}


\title{Robust Group Linkage}

%
%
%
%

\numberofauthors{3}
\author{}

\author{
\newcounter{savecntr}
\newcounter{restorecntr}
\alignauthor Pei Li\\
       \affaddr{University of Zurich}\\
       \email{\small peili@ifi.uzh.ch}
\alignauthor Xin Luna Dong \setcounter{savecntr}{\value{footnote}}\thanks{Research conducted at AT\&T Labs--Research.} \\
\affaddr{Google Inc.}\\
\email{\small lunadong@google.com}
\alignauthor Songtao Guo \setcounter{restorecntr}{\value{footnote}}%
  \setcounter{footnote}{\value{savecntr}}\footnotemark
  \setcounter{footnote}{\value{restorecntr}}\\
\affaddr{LinkedIn Corp.}\\
\email{\small songtao.gg@gmail.com}
\and\alignauthor Andrea Maurino\\
       \affaddr{University of Milan-Bicocca}\\
       \email{\small maurino@disco.unimib.it}
\alignauthor
Divesh Srivastava\\
       \affaddr{AT\&T Labs-Research}\\
       \email{\small divesh@research.att.com}
}
\maketitle

\begin{abstract}
{\small
We study the problem of {\em group linkage}: linking
records that refer to entities in the same group.
Applications for group linkage include finding businesses
in the same chain, finding conference attendees from the
same affiliation, finding players from the same team, etc.
Group linkage faces challenges not present for traditional
record linkage. First, although different members in the
same group can share some similar {\em global} values of an attribute,
they represent different entities so can also have
distinct {\em local} values for the same or different attributes, requiring a high {\em tolerance}
for value diversity. Second, groups can be huge (with tens of thousands of records), requiring high \emph{scalability} even after using good blocking strategies.}

{\small
We present a two-stage algorithm:
the first stage identifies {\em cores} containing
records that are very likely to belong to the same group, while being robust to possible erroneous values;
the second stage collects strong evidence from
the cores and leverages it for merging more records into the
same group, while being tolerant to differences in local values of an attribute. 
Experimental results show the high effectiveness and efficiency of our algorithm
on various real-world data sets.
}
\end{abstract}

\section{Introduction}\label{sec:intro}

{\em Record linkage} aims at linking records that refer to the
same real-world entity and it has been extensively studied
in the past years (surveyed in~\cite{EIV07,KSS06}).
In this paper we study a related but different problem that we
call {\em group linkage}: linking
records that refer to entities in the same group.

One major motivation for our work comes from identifying
{\em business chains}--connected business entities that share
a brand name and provide similar products and services
(\eg, {\em Walmart, McDonald's}).
With the advent of the Web and mobile devices,
we are observing a boom in {\em local search}; that is, searching
local businesses under geographical constraints.
Local search engines include {\em Google Maps, Yahoo! Local,
YellowPages, yelp, ezlocal,} etc.
The knowledge of business chains can have a big economic value
to local search engines, as it
allows users to search by business chain,
allows search engines to render the returned results by chains,
allows data collectors to clean and enrich information within
the same chain, allows the associated review system to connect reviews
on branches of the same chain, and allows sales people to target
potential customers. Business listings are rarely associated with
specific chains explicitly in real-world business-listing collections, so we need to 
identify the chains. Sharing the same name, phone number, or
URL domain name can all serve as evidence of belonging to the
same chain. However, for US businesses alone
there are tens of thousands of chains
and as we show soon, we cannot easily develop any rule set that applies to all chains.

We are also motivated by applications where we need to find
people from the same organization, such as counting
conference attendees from the same affiliation,
counting papers by authors from the same institution,
and finding players of the same team. The organization information is often
missing, incomplete, or simply too heterogeneous to
be recognized as the same (\eg, ``International Business
Machines Corporation'', ``IBM Corp.'', ``IBM'', ``IBM Research Labs'',
``IBM-Almaden'', etc., all refer to the same organization).
Contact phones, email addresses, and mailing addresses of people all
provide extra evidence for group linkage, but they can also vary
for different people even in the same organization.

Group linkage faces challenges not present for traditional
record linkage. First, although different members in the
same group can share some similar {\em global} values of an attribute,
they represent different entities so can also have
distinct {\em local} values for the same or different attributes. For example, different branches
in the same business chain can provide different local
phone numbers, different addresses, etc. It is non-trivial
to distinguish such differences from
various representations for the same value and
sometimes erroneous values in the data. Second,
there are often millions of records for group linkage, and a group
can contain tens of thousands of members.
A good blocking strategy should
put these tens of thousands of records in the same block;
but performing record linkage via traditional pairwise comparisons within such huge blocks can be very expensive.
Thus, {\em scalability} is a big challenge.
We use the following example of identifying business chains
throughout the paper for illustration.

\begin{table}
\vspace{-.1in}
 \scriptsize
  \centering
  \caption{\small Identified top-5 US business chains. For each chain,
  we show the number of stores, distinct business names,
  distinct phone numbers, distinct URL domain names,
  and distinct categories. \label{tbl:chain}}
  \begin{tabular}{|c|c|c|c|c|c|}
  \hline
  {\bf Name} &{\bf \#Store}& {\bf \#Name} & {\bf \#Phn}&{\bf \#URL} & {\bf \#Cat}\\
  \hline
  {\em SUBWAY} & 21,912 & 772& 21, 483& 6& 23\\
  {\em Bank of America} & 21,727 & 48 & 6,573 & 186 & 24\\
  {\em U-Haul} & 21,638 & 2,340& 18,384 & 14& 20\\
  {\em USPS - United State Post Office} & 19,225 & 12,345&5,761 &282 &22 \\
  {\em McDonald's} & 17,289 & 2401& 16,607& 568& 47\\
  \hline
\end{tabular}
\vspace{-.2in}
\end{table}
\begin{table}
 \scriptsize
  \centering
\vspace{-.1in}
\caption{\small Real-world business listings. We show only state for
{\sf location} and simplify names of {\sf category}.
There is a wrong value in italic font.  \label{tab:motiv}}
  \begin{tabular}{c|c|c|c|c|c}
  \hline
  {\sf RID} &{\sf name}& {\sf phone}&{\sf URL} (domain)&{\sf location}&{\sf category}\\
  \hline
  $r_{1}$&Home Depot, The&808&&NJ&furniture\\
  $r_{2}$&Home Depot, The&808&&NY&furniture\\
  $r_{3}$&Home Depot, The&808&homedepot&MD&furniture\\
  $r_{4}$&Home Depot, The&808&homedepot&AK&furniture\\
  $r_{5}$&Home Depot, The&808&homedepot&MI&furniture\\
  $r_{6}$&Home Depot, The&101&homedepot&IN&furniture\\
  $r_{7}$&Home Depot, The&102&homedepot&NY&furniture\\
  $r_{8}$&Home Depot, USA&103&homedepot&WV&furniture\\
  $r_{9}$&Home Depot USA&808&&SD&furniture\\
  $r_{10}$&Home Depot - Tools&808&&FL&furniture\\
\hline
  $r_{11}$&Taco Casa&&tacocasa&AL&restaurant\\
  $r_{12}$&Taco Casa&900&tacocasa&AL&restaurant\\
  $r_{13}$&Taco Casa&900&tacocasa,&AL&restaurant\\
    &&&{\em tacocasatexas}&&\\
  $r_{14}$&Taco Casa&900&&AL&restaurant\\
  $r_{15}$&Taco Casa&900&&AL&restaurant\\
\hline
  $r_{16}$&Taco Casa&701&tacocasatexas&TX&restaurant\\
  $r_{17}$&Taco Casa&702&tacocasatexas&TX&restaurant\\
  $r_{18}$&Taco Casa&703&tacocasatexas&TX&restaurant\\
\hline
  $r_{19}$&Taco Casa&704&&NY&food store\\
\hline
  $r_{20}$&Taco Casa&&tacodelmar&AK&restaurant\\
  \hline
\end{tabular}
\vspace{-.1in}
\end{table}

\vspace{-.1in}
\begin{example}
\label{ex:motivation}
We consider a set of 18M real-world business listings in the
US extracted from {\em Yellowpages.com}, each describing a business
by its name, phone number, URL domain name, location,
and category. Our algorithm automatically finds 600K
business chains and 2.7M listings that belong to
these chains. Table~\ref{tbl:chain} lists the largest five
chains we found. We observe that (1) each chain contains
up to 22K different branch stores, (2) different
branches from the same chain can have
a large variety of names, phone numbers,
and URL domain names, and (3) even chains of similar
sizes can have very different numbers of distinct URLs 
(same for other attributes). Thus, rule-based linkage
can hardly succeed and scalability is a necessity.

Table~\ref{tab:motiv} shows a set of 20 business listings
(with some abstraction) in this data set. After
investigating their webpages manually, we find that $r_1-r_{18}$ belong to
three business chains: $\mbox{Ch}_1=\{r_1-r_{10}\},
\mbox{Ch}_2=\{r_{11}-r_{15}\}$,
and $\mbox{Ch}_3=\{r_{16}-r_{18}\}$; $r_{19}$ and $r_{20}$ do not
belong to any chain. Note the slightly different names
for businesses in chain $\mbox{Ch}_1$; also note that $r_{13}$
is integrated from different sources and contains two URLs,
one ({\em tacocasatexas}) being wrong.

Simple linkage rules do not work well on this data set.
For example, if we require only high similarity on {\sf name} for chain
identification, we may wrongly decide that $r_{11}-r_{20}$
all belong to the same chain as they share a popular
restaurant name {\em Taco Casa}. Traditional linkage strategies
do not work well either. If we apply Swoosh-style
linkage~\cite{WMK+09} and iteratively merge records with high similarity on
{\sf name} and shared {\sf phone} or {\sf URL}, we can
wrongly merge $\mbox{Ch}_2$ and $\mbox{Ch}_3$ because of the
wrong URL from $r_{13}$.
If we require high similarity between listings on
{\sf name, phone, URL, category}, we may either split
$r_6-r_8$ out of chain $\mbox{Ch}_1$
because of their different local phone numbers,
or learn a low weight for {\sf phone} but split
$r_9-r_{10}$ out of chain $\mbox{Ch}_1$ since sharing the
same phone number, the major evidence, is downweighted. \rbox
\end{example}

The key idea in our solution is to find strong evidence
that can glue group members together,
while being tolerant to differences in values specific
for individual group members.
For example, we wish to reward sharing of primary values, such as \emph{primary phone numbers}
or \emph{URL domain names} for chain identification, but would not penalize differences
from local values, such as \emph{locations}, \emph{\emph{local phone numbers}}, and even \emph{\emph{categories}}.
For this purpose, our algorithm proceeds in two stages.
First, we identify {\em cores} containing records
that are very likely to belong to the same group.
Second, we collect strong evidence from the resulting cores,
such as primary phone numbers and URL domain names in business chains,
based on which we cluster the cores and remaining
records into groups.
The use of cores and strong evidence distinguishes our
clustering algorithm from traditional clustering techniques
for record linkage. In this process, it is crucial that
core generation makes very few false positives even in the presence
of erroneous values, such that we can avoid ripple
effect on clustering later.
Our algorithm is designed to ensure efficiency and scalability.

The group linkage problem we study in this paper
is different from the group linkage in~\cite{10.1109/ICSC.2010.26, 4221698},
which decides similarity between \emph{pre-specified} groups of records.
Our goal is to find records that belong to the same group
and we make three contributions.
\begin{enumerate}\tightlist
  \item We study core generation in presence of erroneous
    data. Our core is {\em robust} in the sense that even if we remove
    a few possibly erroneous records from a core, we still have
    strong evidence that the rest of the records in the core must belong to
    the same group. 
  \item We then reduce the group linkage problem into
    clustering cores and remaining records. Our clustering
    algorithm leverages strong evidence collected from cores and
    meanwhile is {\em tolerant} to value variety of records in the same group.
  \item We conducted experiments on two real-world data sets
    in different domains,
    showing high efficiency and effectiveness of our algorithms.
\eat{
  containing
    6.8 million of business listings; our algorithm finished in
    2.4 hours and obtained a precision and recall of over .95,
    improving over traditional record linkage techniques,
    which obtained a precision of .8 and a recall of .5.}
\end{enumerate}

Note that we assume prior to group linkage, we first conduct
record linkage (\eg, \cite{DBLP:journals/pvldb/GuoDSZ10}).  
Our experiments show that minor mistakes for record linkage do not significantly affect the results of
group linkage, and records that describe the same entity
but fail to be merged in the record-linkage step are often put into the same group.
We plan to study how to combine record linkage and group linkage
to improve the results of both in the future.

In the rest of the paper, Section~\ref{sec:related} discusses related
work. Section~\ref{sec:problem} defines the problem
and provides an overview of our solution.
Sections~\ref{sec:core}-\ref{sec:sate}
describe the two stages in our solution.
Section~\ref{sec:experiment} describes
experimental results.
Section~\ref{sec:conclude} concludes.

\section{Related Work}\label{sec:related}
Record linkage has been extensively studied in the past
(surveyed in~\cite{EIV07, KSS06}). Traditional linkage
techniques aim at linking records that refer to the same
real-world entity, so implicitly assume value consistency
between records that should be linked. Group linkage
is different in that it aims at linking records that
refer to different entities in the same group. The variety
of individual entities requires better use of
strong evidence and tolerance on different
values even within the same group. These two features differentiate
our work from any previous linkage technique.

For record clustering in linkage, existing work
may apply the transitive rule~\cite{Hernandez98real-worlddata},
or do match-and-merge~\cite{WMK+09},
or reduce it to an optimization problem~\cite{Hassanzadeh09frameworkfor}. 
Our work is different in that our core-identification algorithm
aims at being robust to a few erroneous records; and our
clustering algorithm emphasizes leveraging the
strong evidence collected from the cores.

For record-similarity computation, existing work
can be rule based\\~\cite{Hernandez98real-worlddata},
classification based~\cite{fellegi69},
or distance based~\cite{Dey:2008:EMH:1326361.1326464}.
There has also been work on weight (or model) learning
from labeled data~\cite{fellegi69, Winkler02methodsfor}.
Our work is different in that in addition to learning a 
weight for each attribute, we also learn a weight for each value
based on whether it serves as important evidence for the group.
Note that some previous works are also tolerant to different
values but leverage evidence that may not be available in our contexts:
\cite{Fan:2009:RRM:1687627.1687674} is tolerant to schema 
heterogeneity from different relations by specifying matching rules;
\cite{DBLP:journals/pvldb/GuoDSZ10} is tolerant to possibly
false values by considering agreement between different data providers;
\cite{LDMS11} is tolerant to out-of-date
values by considering time stamps;
we are tolerant to diversity within the same group.

Two-stage clustering has been proposed in the IR and machine learning
community~\cite{DBLP:conf/vldb/BansalCKT07, Larsen:1999:FET:312129.312186, 
Liu:2002:DCC:564376.564411, Clustering07ricochet:a, Yoshida:2010:PND:1835449.1835454}; however, they identify cores in different ways.
Techniques in~\cite{Larsen:1999:FET:312129.312186, Clustering07ricochet:a}
consider a core as a single record, either randomly selected or
selected according to the weighted degrees of nodes in the graph.
Techniques in~\cite{Yoshida:2010:PND:1835449.1835454} generate
cores using agglomerative clustering but can be too conservative
and miss strong evidence. Techniques in~\cite{DBLP:conf/vldb/BansalCKT07}
identify cores as {\em bi-connected components}, where
removing any node would not disconnect the graph.
Although this corresponds to the {\em 1-robustness}
requirement in our solution (defined in Section~\ref{sec:core}),
they generate overlapping clusters; it is not obvious how to
derive non-overlapping clusters in applications such as business-chain
identification and how to extend their techniques to guarantee
$k$-robustness.
Finally, techniques in~\cite{Larsen:1999:FET:312129.312186, Liu:2002:DCC:564376.564411}
require knowledge of the number of clusters for one of the stages,
so do not directly apply in our context.
We compare with these methods whenever applicable in experiments
(Section~\ref{sec:experiment}), showing that our algorithm is
robust in presence of erroneous values and consistently generates
high-accuracy results on data sets with different features.

Finally, we distinguish our work from the {\em group linkage} 
in~\cite{10.1109/ICSC.2010.26, 4221698}, which has different goals.
On et al.~\cite{4221698} decided similarity between pre-specified groups of records
and the group-entity relationship is many-to-many
(\eg, authors and papers). Huang~\cite{10.1109/ICSC.2010.26} decided whether
two pre-specified {\em groups} of records from different data sources
refer to the same group by analysis of social network. Our goal is 
to find records that belong to the same group.

\section{Overview} \label{sec:problem}
This section formally defines the group linkage problem and
provides an overview of our solution.

\subsection{Problem definition}
Let $\bf R$ be a set of records that describe real-world entities by a set of attributes $\bf A$.
For each record $r \in \bf R$, we 
denote by $r.A$ its value on attribute $A \in \bf A$.
Sometimes a record may contain erroneous or missing values.

We consider the {\em group linkage} problem;
that is, finding records that represent entities
belonging to the same real-world group. As an example application, we wish to find
{\em business chains}--a set of business entities
with the same or highly similar names that provide similar products and services (e.g., {\em Walmart}, {\em Home Depot}, {\em Subway} and {\em McDonald's}).\footnote{\small http://en.wikipedia.org/wiki/Chain$\_$store.}
We focus on non-overlapping groups, which often hold in applications.

\begin{definition}[Group linkage]
Given a set $\bf R$ of records, {\em group linkage} identifies a set of clusters $\bf CH$ of records
in $\bf R$, such that (1) records that represent real-world entities
in the same group belong to one cluster,
and (2) records from different groups belong to different clusters. \rbox
\end{definition}

\begin{example}
Consider records in Example~\ref{ex:motivation}, where each record describes a business store (at a distinct location) by attributes {\sf name}, {\sf phone}, {\sf URL}, {\sf location}, and {\sf category}.

The ideal solution to the group linkage problem contains 5 clusters: $\mbox{Ch}_1=\{r_1-r_{10}\}$, $\mbox{Ch}_2=\{r_{11}-r_{15}\}$,
$\mbox{Ch}_3=\{r_{16}-r_{18}\}$, $\mbox{Ch}_4=\{r_{19}\}$,
and $\mbox{Ch}_5=\{r_{20}\}$.
Among them, $\mbox{Ch}_2$ and $\mbox{Ch}_3$ represent two different chains
with the same name. \rbox
\end{example}

\subsection{Overview of our solution}
Group linkage is related to but different from traditional
record linkage because it essentially looks for records that
represent entities in the same group, 
rather than records that represent exactly the same
entity. Different members in the same
group often share a certain amount of commonality
(\eg, common name, primary phone, and URL domain of chain stores), but meanwhile
can also have a lot of differences
(\eg, different addresses, local phone numbers, and local URL domains);
thus, we need to allow much higher
variety in some attribute values to avoid false negatives.
On the other hand, as we have shown
in Example~\ref{ex:motivation}, simply lowering our requirement
on similarity of records or similarity of a few attributes
in clustering can lead to a lot of false positives.

The key intuition of our solution is to distinguish between {\em strong}
evidence and {\em weak} evidence. For example, different branches
in the same business chain often share the same URL domain name
and those in North America often share the same 1-800 phone
number. Thus, a URL domain or phone number shared among many
business listings with highly similar names can serve as strong
evidence for chain identification. In contrast,
a phone number shared by only a couple of business entities
is much weaker evidence, since one might be an erroneous or
out-of-date value. 

To facilitate leveraging strong evidence, our solution consists of two stages.
The first stage collects records that are highly likely to
belong to the same group; for example, a set of business listings
with the same name and phone number are very likely to
be in the same chain. We call the
results {\em cores} of the groups; from them we can collect
strong evidence such as name, primary phone number,
and primary URL domain of chains. The key goal of this stage is to be robust
against erroneous values and make as few false positives as possible,
so we can avoid identifying strong evidence wrongly and
causing incorrect ripple effect later; however, we need to keep in mind that
being too strict can miss important strong evidence.

The second stage clusters cores and remaining records into groups
according to the discovered strong evidence.
It decides whether several cores belong to the same group,
and whether a record that does not belong to any core actually
belongs to some group. It also employs weak evidence, but treats it
differently from strong evidence.
The key intuition of this stage is to leverage
the strong evidence and meanwhile be tolerant to
diversity of values in the same group, so we can reduce false
negatives made in the first stage.

We next illustrate our approach for business-chain identification.
\begin{example}\label{eg:illu}
Continue with the motivating example. In the first stage we
generate three cores: $\mbox{Cr}_1=\{r_{1}-r_{7}\},
\mbox{Cr}_2=\{r_{14},r_{15}\},
\mbox{Cr}_3=\{r_{16}-r_{18}\}$. Records $r_{1}-r_{7}$ are in the same core
because they have the same name, five of them ($r_{1}-r_{5}$)
share the same phone number {\em 808} and five of them ($r_{3}-r_{7}$)
share the same URL {\em homedepot}.
Similar for the other two cores. Note that
$r_{13}$ does not belong to any core, because one of its URLs
is the same as that of $r_{11}-r_{12}$, and one is the same as that of $r_{16}-r_{18}$,
but except name, there is no other common information between these two
groups of records. To avoid mistakes,
we defer the decision on $r_{13}$. Indeed, recall that
{\em tacocasatexas} is a wrong value for $r_{13}$.
For a similar reason, we defer the decision on $r_{12}$.

In the second stage, we generate groups--business chains.
We merge $r_{8}-r_{10}$ with core $\mbox{Cr}_1$, because
they have similar names and share either the primary phone number or the primary URL.
We also merge $r_{11}-r_{13}$ with core $\mbox{Cr}_2$, because
(1) $r_{12}-r_{13}$ share the primary phone {\em 900} with $\mbox{Cr}_2$, and
(2) $r_{11}$ shares the primary URL {\em tacocasa} with $r_{12}-r_{13}$.
We do not merge $\mbox{Cr}_2$ and $\mbox{Cr}_3$ though, because
they share neither the primary phone nor the primary URL.
We do not merge $r_{19}$ or $r_{20}$ to any core, because
there is again not much strong evidence. We thus obtain the
ideal result. \rbox
\end{example}

To facilitate this two-stage solution, we find attributes
that provide evidence for group identification and
classify them into three categories. 
\begin{itemize}\tightlist
  \item {\em Common-value attribute:} We call an attribute $A$
a common-value attribute if all entities
in the same group have the same or highly similar $A$-values.
Such attributes include {\sf business-name} for chain identification
and {\sf organization} for organization linkage.
  \item {\em Dominant-value attribute:} We call an attribute $A$
a dominant-value attribute if entities in the
same group often share one or a few primary $A$-values
(but there can also exist other less-common values),
and these values are seldom used by entities outside the group.
Such attributes include {\sf phone} and {\sf URL-domain}
for chain identification, and {\sf office-address},
{\sf phone-prefix}, and {\sf email-server} for organization linkage.
  \item {\em Multi-value attribute:} We call the rest of
the attributes mutli-value attributes as there is often
a many-to-many relationship between groups and values of
these attributes. Such attributes include
{\sf category} for chain identification.
\end{itemize}
\noindent
The classification can be either learned from training data 
based on cardinality of attribute values, or 
performed by domain experts since there are typically
only a few such attributes.

We describe core identification
in Section~\ref{sec:core} and group linkage
in Section~\ref{sec:sate}. Our algorithms require common-value
and dominant-value attributes, which typically exist for groups
in practice. While we present the algorithms for the setting of one machine, a lot of components of our algorithms can be easily parallelized in Hadoop infrastructure~\cite{Shvachko:2010:HDF:1913798.1914427, Chambers:2010:FEE:1806596.1806638}; it is not the focus of the paper and we briefly describe the opportunities in Section~\ref{sec:sum}. 



\section{Core Identification}\label{sec:core}
The first stage of our solution creates cores consisting of records
that are very likely to belong to the same group.
The key goal in core identification is to be robust to possible erroneous
values. This section starts with presenting the
criteria we wish the cores to meet (Section~\ref{sub-sec:def.}),
then describes how we efficiently construct similarity graphs
to facilitate core finding (Section~\ref{sub-sec:graph}),
and finally gives the algorithm for core identification
(Section~\ref{sub-sec:split}). Note that the notations in this section 
can be slightly different from those in Graph Theory.

\subsection{Criteria for a core}\label{sub-sec:def.}
At the first stage we wish to make only decisions that are highly
likely to be correct; thus, we require that each core contains only
highly similar records, and different cores are fairly different and
easily distinguishable from each other.
In addition, we wish that our results are robust even in the presence
of a few erroneous values in the data.
In the motivating example, $r_{1}-r_{7}$ form a good core,
because {\em 808} and {\em homedepot}
are very popular values among these records.
In contrast, $r_{13}-r_{18}$ do not form a good core,
because records $r_{14}-r_{15}$ and $r_{16}-r_{18}$ do not share
any phone number or URL domain; the only ``connector''
between them is $r_{13}$,
so they can be wrongly merged if $r_{13}$ contains erroneous values.
Also, considering $r_{13}-r_{15}$ and $r_{16}-r_{18}$
as two different cores is risky, because
(1) it is not very clear whether $r_{13}$ is in the same chain
as $r_{14}-r_{15}$ or as $r_{16}-r_{18}$, and (2) these two cores
share one URL domain name so are not fully distinguishable.

We capture this intuition with {\em connectivity of
a similarity graph}. We define the {\em similarity graph} of
a set $\bf R$ of records as an undirected graph,
where each node represents a record in {\bf R}, and an edge
indicates high similarity between the connected records
(we describe later what we mean by {\em high similarity}).
Figure~\ref{fig:graph} shows the similarity
graph for the motivating example.

Each core would correspond to a connected sub-graph
of the similarity graph. We wish such a sub-graph to be {\em robust}
such that even if we remove a few nodes the sub-graph is still
connected; in other words, even if there are some erroneous
records, without them we
still have enough evidence showing that the rest of the records
should belong to the same group. The formal definition goes
as follows.

\begin{definition}[$k$-robustness]
A graph $G$ is \emph{$k$-robust} if after removing arbitrary $k$ nodes
and edges to these nodes, $G$ is still connected.
A clique or a single node is $k$-robust for any $k$. \rbox
\end{definition}

In Figure~\ref{fig:graph}, the subgraph with nodes $r_{1}-r_{7}$ is
2-robust. That with $r_{11}-r_{18}$ is not 1-robust, as removing
$r_{13}$ can disconnect it.

According to the definition, we can partition the similarity graph
into a set of $k$-robust subgraphs. As we do not wish to split
any core unnecessarily, we require the {\em maximal $k$-robust
partitioning}: 

\begin{definition}[Maximal $k$-robust partitioning]\label{def:core-clustering}
Let $G$ be a similarity graph. A partitioning of $G$ is
a {\em maximal $k$-robust partitioning} if it
satisfies the following properties.
\begin{enumerate}\tightlist
\item Each node belongs to one and only one partition.
\item Each partition is $k$-robust.
\item The result of merging any partitions is not $k$-robust. \rbox
\end{enumerate}
\end{definition}

Note that a data set can have more than one maximal $k$-robust partitioning. Consider $r_{11}-r_{18}$ in Figure~\ref{fig:graph}.
There are three maximal 1-robust partitionings:
$\{\{r_{11}\}, \{r_{12}, r_{14}-r_{15}\}, \{r_{13}, r_{16}-r_{18}\}\}$;
$\{\{r_{11}-r_{12}\}, \{r_{14}-r_{15}\}, \{r_{13}, r_{16}-r_{18}\}\}$;
and $\{\{r_{11}-r_{15}\}, \{r_{16}-r_{18}\}\}$.
If we treat each partitioning as a possible world, records that
belong to the same partition in all possible worlds have high
probability to belong to the same group and so form a core.
\eat{In these partitionings, $r_{13}$ is merged either with $r_{16}-r_{18}$
or with $r_{11}, r_{12}, r_{14}, r_{15}$. We wish to defer decision on $r_{13}$
so we exclude it from any core.
}Accordingly, we define a core as follows and can prove its 
$k$-robustness.
\begin{definition}[$k$-Core]
Let $\bf R$ be a set of records and $G$ be the similarity
graph of $\bf R$. The records that
belong to the same subgraph in every maximal $k$-robust partitioning
of $G$ form a {\em $k$-core} of $\bf R$. A core contains at least 2
records.
\rbox
\end{definition}
\begin{property}
A $k$-core is $k$-robust. \rbox
\end{property}

\begin{proof}
If a $k$-core $C_{r}$ of $G$ is not $k$-robust, there exists a maximal $k$-robust partitioning in $G$, where two nodes $r$ and $r'$ in $C_r$ are in different partitions of this partitioning (proved by Lemma~\ref{lem:separator}). This conflicts with the fact that records in $C_r$ belong to the same partition in every maximal $k$-robust partitioning of $G$. Therefore, a $k$-core is $k$-robust.
\end{proof}

\begin{figure}[t]
\centering
\vspace{-0.1in}
\includegraphics[scale=.35]{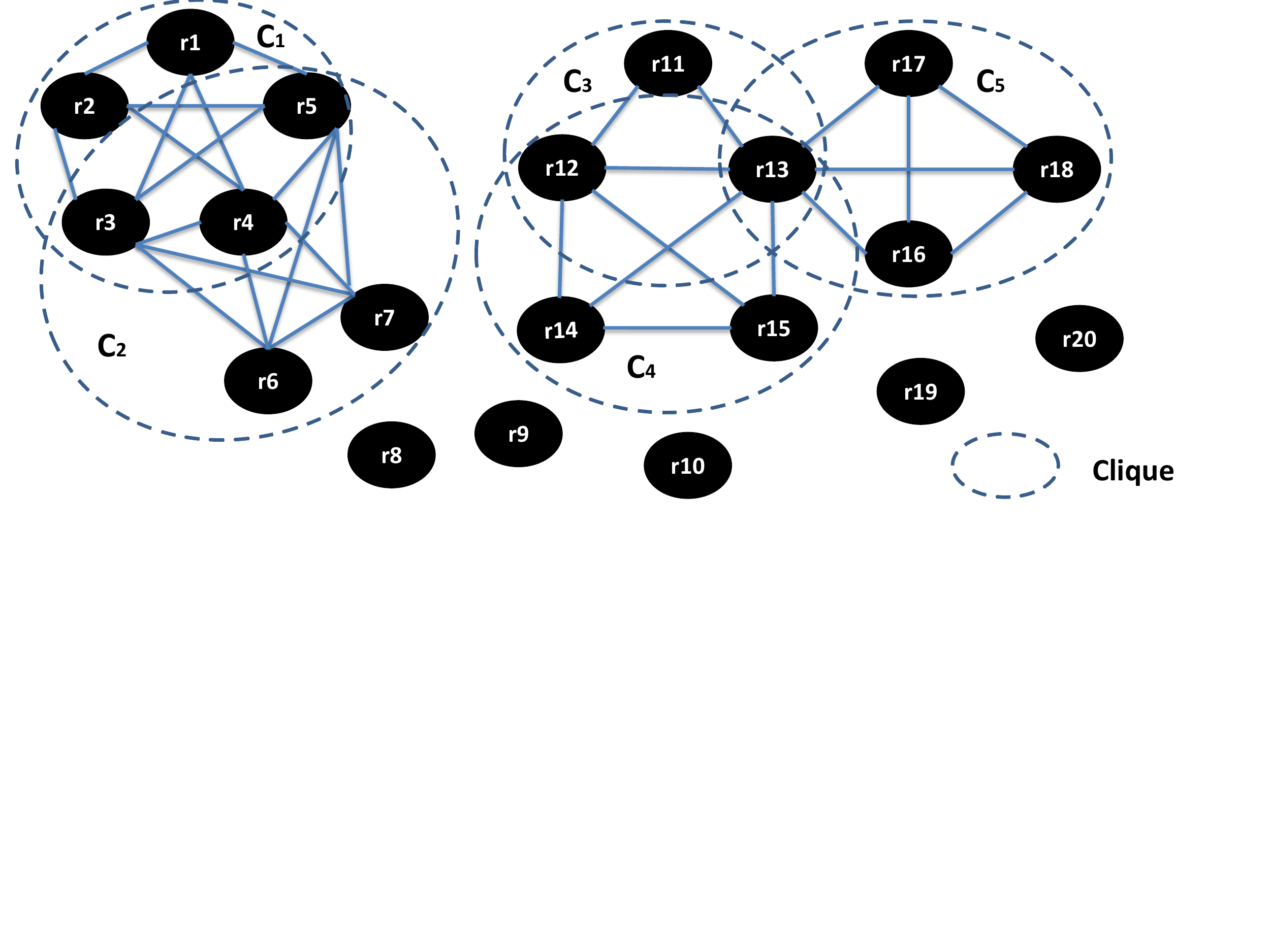}
\vspace{-1.5in} {\small\caption{\label{fig:graph}Similarity graph for
records in Table~\ref{tab:motiv}.}}
\vspace{-.15in}
\end{figure}
\vspace{-.2in}
\begin{example}
\label{ex:core}
Consider Figure~\ref{fig:graph} and assume $k=1$.
There are two connected sub-graphs.
For records $r_{1}-r_{7}$, the subgraph is 1-robust, so they form a $1$-core.
For records $r_{11}-r_{18}$, there are three maximal 1-robust
partitionings for the subgraph, as we have shown.
Two subsets of records belong to the same subgraph in each partitioning:
$\{r_{14}-r_{15}\}$ and $\{r_{16}-r_{18}\}$; they form 2 $1$-cores.\rbox
\end{example}

\begin{table}
\vspace{-.1in}
 \scriptsize
  \centering
  \caption{Simplified inverted index for
the similarity graph in Figure~\ref{fig:graph}.\label{tbl:inverted}}
 \begin{tabular}{|c|l|c|}
\hline
Record & \multicolumn{1}{|c|}{V-Cliques} & Represent\\
\hline
$r_{1/2}$ & $C_1$ & $r_{1}-r_{2}$\\
$r_{3}$ & $C_1, C_2$ & $r_{3}$ \\
$r_{4}$ & $C_1, C_2$ & $r_{4}$ \\
$r_{5}$ & $C_1, C_2$ & $r_{5}$ \\
$r_{6/7}$ & $C_2$ & $r_{6}-r_{7}$ \\
$r_{11}$ & $C_3$ & $r_{11}$\\
$r_{12}$ & $C_3, C_4$ & $r_{12}$\\
$r_{13}$ & $C_3, C_4, C_5$ & $r_{13}$\\
$r_{14/15}$ & $C_4$ & $r_{14}-r_{15}$\\
$r_{16/17/18}$ & $C_5$ & $r_{16}-r_{18}$ \\
\hline
\end{tabular}
 \vspace{-.2in}
\end{table}
%
%
%
%
%
\subsection{Constructing similarity graphs}\label{sub-sec:graph}
Generating the cores requires analysis on the similarity graph.
Even after blocking, a block can contain tens of thousands of records,
so it is not scalable to compare every pair of records in the same block
and create edges accordingly.
We next describe how we construct and represent
the similarity graph in a scalable way.

We add an edge between two records if
they have the same value for each
common-value attribute and share at least one value
on a dominant-value attribute\footnote{\small In practice,
we require only highly similar values for common-value attributes
and apply the transitive rule on similarity (\ie, if $v_1$ and
$v_2$ are highly similar, and so are $v_2$ and $v_3$,
we consider $v_1$ and $v_3$ highly similar).};
our experiments show advantages of this method over
other edge-adding strategies (Section~\ref{sec:exp-core}).
All records that share values on the common-value attributes
and share the same value on a dominant-value attribute
form a clique, which we call a {\em v-clique}.
We can thus represent the graph with a set of v-cliques,
denoted by $\bf C$; for example,
the graph in Figure~\ref{fig:graph}
can be represented by 5 v-cliques ($C_1-C_5$).
In addition, we maintain an {\em inverted index} $\bar L$, where each entry
corresponds to a record $r$ and contains the v-cliques that $r$
belongs to. 
Whereas the size of the
similarity graph can be quadratic in the number of the nodes,
the size of the inverted index is only linear in that number.
The inverted index also makes it easy to find {\em adjacent v-cliques}
(\ie, v-cliques that share nodes), as they appear in the same entry.

Graph construction is then reduced to v-clique finding,
which can be done by scanning values of dominant-value
attributes. In this process, we wish to prune a v-clique if it is
a sub-clique of another one. Pruning by checking every pair
of v-cliques can be very expensive since the number
of v-cliques is also huge. Instead, we do it together with v-clique finding.
Specifically, our algorithm {\sc GraphConstruction} takes
$\bf R$ as input and outputs $\bf C$ and $\bar L$.
We start with ${\bf C}=\bar L=\emptyset$.
For each value $v$ of a dominant-value attribute,
we denote the set of records with $v$ by $\bar R_v$ and
do the following.

\begin{enumerate}\tightlist
\item Initialize the v-cliques for $v$ as $\bf C_v=\emptyset$.
Add a single-record cluster for each record $r\in\bar R_v$
to a working set $\bar T$. Mark each cluster as ``unchanged''.

\item For each $r \in \bar R_v$, scan $\bar L$ and consider
each v-clique $C \in \bar L(r)$ that has not been considered yet.
For all records in $C\cap R_v$, merge their clusters.
Mark the merged cluster as ``changed'' if the result is not a
proper sub-clique of $\bar C$.
If $C \subseteq \bar R_v$, remove $C$ from $\bf C$.
This step removes the v-cliques that must be sub-cliques of
those we will form next.

\item For each cluster $C \in \bar T$, if there exists $C' \in \bf C_v$
such that $C$ and $C'$ share the same value
for each common-value attribute, remove $C$ and $C'$ from
$\bar T$ and $\bf C_v$ respectively, add $C\cup C'$ to $\bar T$
and mark it as ``changed'';
otherwise, move $C$ to $\bf C_v$. This step merges clusters that
share values on common-value attributes. At the end,
$\bf C_v$ contains the v-cliques with value $v$.

\item Add each v-clique with mark ``changed'' in $\bf C_v$
to $\bf C$ and update $\bar L$ accordingly. The marking
prunes size-1 v-cliques and the
sub-cliques of those already in $\bf C$.
\end{enumerate}

\begin{proposition}\label{pro:graphconstruct}
Let $\bf R$ be a set of records. Denote by
$n(r)$ the number of values on dominant-value attributes
from $r \in \bf R$. Let $n=\sum_{r \in \bf R}n(r)$ and
$m=\max_{r \in \bf R}n(r)$. Let $s$ be the maximum v-clique size.
Algorithm {\sc GraphConstruction} (1) runs in time $O(ns(m+s))$,
(2) requires space $O(n)$,
and (3) its result is independent of the order in which we consider
the records. \rbox
\end{proposition}

\begin{proof}
\label{proof:graph}
We first prove that {\sc GraphConstruction} runs in time $O(ns(m+s))$. Step 2 of the algorithm takes in time $O(nsm)$, where it takes in time $O(ns)$ to scan all records for a dominant-value attribute, and a record can be scanned maximally $m$ times. Step 3 takes in time $O(ns^2)$. Thus, the algorithm runs in time $O(ns(m+s))$.

We next prove that {\sc GraphConstruction} requires space $O(n)$. For each value $v$ of a dominate-value attribute, the algorithm keeps three data sets: $\bar L$ that takes in space $O(n)$, $\bf C_v$ and $\bar T$ that require space in total no greater than $O(\bf |R|)$. Since $O(n)\geq O(\bf |R|)$, the algorithm requires space $O(n)$.

We now prove that the result of {\sc GraphConstruction} is order independent. Given $\bar L$ and $\bar R_v$, Step 2 scan $\bar L$ and apply transitive rule to merge clusters of records in $C \cap \bar R_v$, for each v-clique $C \in \bar L$. The process is independent from the order in which we consider the records in $\bar R_v$. The order independence of the result in Step 3 is proven in~\cite{DBLP:journals/vldb/BenjellounGMSWW09}. Therefore, the final result is independent from the order in which we consider the records.
\end{proof}


%
%
%
%
\begin{example}\label{ex:graph}
Consider graph construction for records in
Table~\ref{tab:motiv}. Figure\ref{fig:graph} shows the
similarity graph and Table~\ref{tbl:inverted}(a)
shows the inverted list.
We focus on records $r_1-r_8$ for illustration.

First, $r_{1}-r_{5}$ share the same name and phone number {\em 808},
so we add v-clique $C_1=\{r_{1}-r_{5}\}$ to $\bf C$.
Now consider URL {\em homedepot} where $\bar R_v=\{r_{3}-r_{8}\}$.
Step 1 generates 6 clusters, each marked ``unchanged'', and
$\bar T = \{\{r_{3}\}, \dots, \{r_{8}\}\}$.
Step 2 looks up $\bar L$ for each record in $\bar R_v$.
Among them, $r_{3}-r_{5}$ belong to v-clique $C_1$, so it
merges their clusters and marks the result $\{r_3-r_5\}$ ``unchanged''
($\{r_3-r_5\} \subset C_1$); then,
$\bar T = \{\{r_{3}-r_{5}\}, \{r_{6}\}, \{r_{7}\}, \{r_8\}\}$.
Step 3 compares these clusters and merges the first three
as they share the same name, marking the result as ``changed''.
At the end, ${\bf C}_v=\{\{r_{3}-r_{7}\}, \{r_8\}\}$.
Finally, Step 4 adds $\{r_{3}-r_{7}\}$ to $\bf C$
and discards $\{r_8\}$ since it is marked ``unchanged''. \rbox
\end{example}

Given the sheer number of records in $\bf R$,
the inverted index can still be huge.
In fact, according to the following theorem, records in the same v-clique
but not any other v-clique must belong to the same core,
so we do not need to distinguish them.
Thus, we simplify the inverted index such that
for each v-clique we keep only a {\em representative}
for nodes belonging only to this v-clique.
Table~\ref{tbl:inverted} shows the simplified index for
the similarity graph in Figure~\ref{fig:graph}.

\begin{theorem}\label{thm:simplify}
Let $G$ be a similarity graph and $G'$ be a graph derived from $G$
by merging nodes that belong to only one and the same $v$-clique. 
Two nodes belong to the same core of $G'$ if and only if
they belong to the same core of $G$. \rbox
\end{theorem}

\begin{proof}
\label{proof:simplify}
We need to prove that (1) if two nodes $r$ and $r'$ belong to the same core in $G'$, they are in the same core of $G$, and (2) if two nodes $r$ and $r'$ belong to the same core of $G$, they are in the same core of $G'$. 

We first prove that if two nodes $r$ and $r'$ belong to the same core in $G'$, they are in the same core of $G$. Suppose there does not exist any core in $G$ that contains both $r$ and $r'$. It means that there exists a maximal $k$-robust partitioning in $G$, where $r$ and $r'$ are in different partitions. Let $P$ be such a partitioning of $G$ and we consider partitioning $P'$ of $G'$, where each pair of nodes in the same partition $C$ of $P$ are in the same partition $C'$ of $P'$ and vice versa. We prove that $P'$ is a maximal $k$-robust partitioning in $G'$. (1) It is obvious that each node in $P'$ belongs to one and only one partition. (2) For each partition $C'$ in $P'$, removing any $k$ nodes in $C'$ is equivalent to removing $n+m$ nodes in $C$, where $n$ nodes belong to more than one $v$-cliques in $C$, $m$ nodes belong to single $v$-cliques in $C$, and $n \leq k$. Since removing $m$ nodes that belong to single $v$-cliques do not disconnect $C$ and we know $n \leq k$, removing the $n+m$ nodes does not disconnect $C$. It in turn proves that removing $k$ nodes in $C'$ does not disconnect $C'$, and $C'$ is $k$-robust. (3) Similarly, we have that the result of of merging any partitions in $P'$ is not $k$-robust. Therefore, $P'$ is a maximal $k$-robust partitioning in $G'$. Given that $r$ and $r'$ are in different partitions of $P'$, there does not exist a core of $G'$ that contains both $r$ and $r'$. This conflicts with the fact that $r$ and $r'$ belong to the same core in $G'$, and further proves that $r$ and $r'$ are in the same core of $G$. 

We next prove that if two nodes $r$ and $r'$ belong to the same core of $G$, they are in the same core of $G'$. Suppose there does not exist any core in $G'$ that contains both $r$ and $r'$. It means that there exists a maximal $k$-robust partitioning in $G'$, where $r$ and $r'$ are in different partitions. Let $P'$ be such a partitioning of $G'$ and we consider partitioning $P$ of $G$, where each pair of nodes in the same partition $C'$ of $P'$ are in the same partition $C$ of $P$ and vice versa. In similar ways as above, we have that $P$ is a maximal $k$-robust partitioning in $G$. Given that $r$ and $r'$ are in different partitions of $P$, there does not exist a core of $G$ that contains both $r$ and $r'$. This conflicts with the fact that $r$ and $r'$ belong to the same core in $G$, and further proves that $r$ and $r'$ are in the same core of $G'$.
%
%
%
\end{proof}

\eat{
\begin{proposition}\label{prop:simplify}
Let $G$ be a similarity graph.
Let $r$ and $r'$ be two nodes in $G$ that belong to the same
v-clique $C$ and not any other v-clique. Then, $r$ and $r'$ must
belong to the same partition in any maximal $k$-robust partitioning.\rbox
\end{proposition}
}


\smallskip
\noindent
{\bf Case study:} On a data set with 18M records (described in Section~\ref{sec:experiment}), our graph-construction algorithm finished in 1.9 hours. The original similarity graph contains 18M nodes and 4.2B edges. The inverted index is of size 89MB, containing 3.8M entries, each associated with at most 8 v-cliques; in total there are 1.2M v-cliques. The simplified inverted index is of size 34MB, containing 1.5M entries, where an entry can represent up to 11K records. Therefore, the simplified inverted index reduces the size of the similarity graph by 3 orders of magnitude.

\subsection{Identifying cores}\label{sub-sec:split}
We solve the core-identification problem by reducing it to a Max-flow/Min-cut Problem.
However, computing the max flow for a given graph $G$
and a source-destination pair
takes time $O(|G|^{2.5})$, where $|G|$ denotes the
number of nodes in $G$; even the simplified inverted index can still contain
millions of entries, so it can be very expensive.
We thus first merge certain v-cliques according to a sufficient
(but not necessary) condition for $k$-robustness
and consider them as a whole in core identification;
we then split the graph into subgraphs according to a necessary
(but not sufficient) condition for $k$-robustness. We apply reduction only on
the resulting subgraphs, which are substantially smaller as
we show at the end of this section. Section~\ref{sub-sec:prun} describes
screening before reduction, Section~\ref{subsubsec:split}
describes the reduction, and Section~\ref{subsubsec:core-detection}
gives the full algorithm, which iteratively applies screening
and the reduction. 

\begin{figure}[t]
\centering
\includegraphics[scale=.35]{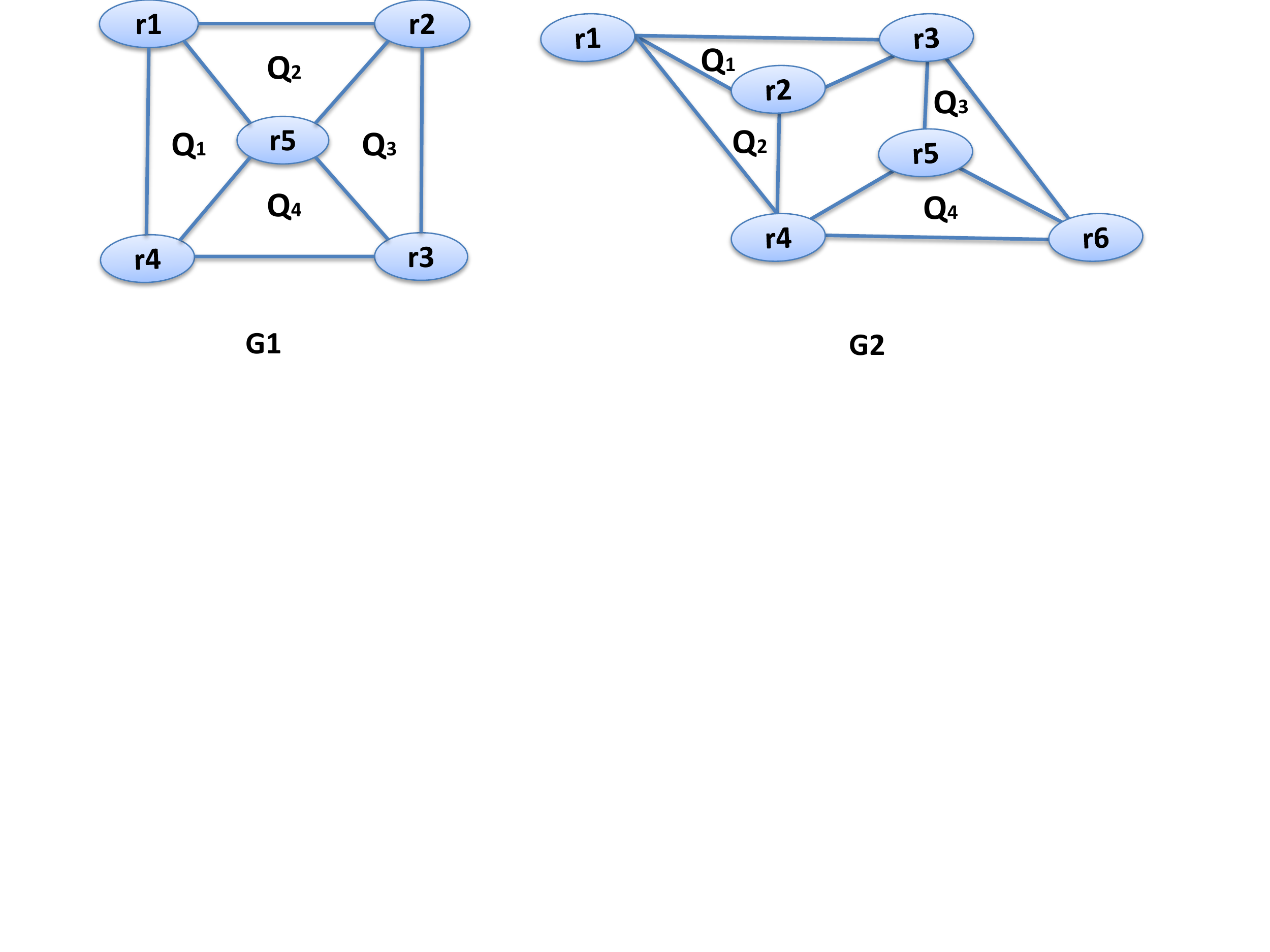}
\vspace{-1.9in}
{\small\caption{Two example graphs.\label{fig:merge-loop}}}
\vspace{-.15in}
\end{figure}
\subsubsection{Screening}\label{sub-sec:prun}
A graph can be considered as a union of
v-cliques, so essentially we need to decide if a
union of v-cliques is $k$-robust.
First, we can prove the following sufficient condition for
$k$-robustness.

\begin{theorem}[$(K+1)$-connected condition]\label{thm:union}
Let $G$ be a graph consisting of a union $Q$ of v-cliques.
If for every pair of v-cliques $C, C' \in Q$, there is a path
of v-cliques between $C$ and $C'$ and every pair of adjacent
v-cliques on the path share at least $k+1$ nodes,
graph $G$ is $k$-robust.\rbox
\end{theorem}

\begin{proof}
\label{proof:k-connected}
Given Menger's Theorem~\cite{Bruhn:2005:MTI:1384589.1384594}, graph $G$ is $k$-robust if for any pair of nodes $r, r'$ in $G$, there exists at least $k+1$ independent paths that do not share any nodes other than $r, r'$ in $G$. We now prove that for any pair of nodes $r, r'$ in graph $G$ that satisfies $(k+1)$-connected condition, there exists at least $k+1$ independent paths between $r, r'$. We consider two cases, 1) $r, r'$ are adjacent such that there exists a v-clique in $G$ that contains $r, r'$; 2) $r, r'$ are not adjacent such that there exists no v-clique in $G$ that contains $r, r'$.

We first consider Case 1 where there exists a v-clique $C$ containing $r, r'$. Since each v-clique in $G$ has more than $k+1$ nodes, there exist at least $k$ 2-length paths and one 1-length path between $r, r' \in C$. It proves that there exists at least $k+1$ independent paths between $r$ and $r'$.

We next consider Case 2 where there exists no v-clique containing $r, r'$ in $G$. Suppose $r\in C, r'\in C'$, where $C, C'$ are different v-cliques in $G$. Since there exists a path of v-cliques between $C$ and $C'$ where every pair of adjacent v-cliques in the path share at least $k+1$ nodes, there exists at least $k+1$ independent paths between $r$ and $r'$.

Given the above two cases, we have that there exist at least $k+1$ independent paths between every pair of nodes in $G$, therefore $G$ is $k$-robust.
\end{proof}

We call a single v-clique or
a union of v-cliques that satisfy the $(k+1)$-connected
condition a {\em $(k+1)$-connected v-union}.
A $(k+1)$-connected v-union must be $k$-robust but
not vice versa.
In Figure~\ref{fig:graph}, subgraph $\{r_{1}-r_{7}\}$ is
a 3-connected v-union, because the only two v-cliques,
$C_1$ and $C_2$, share 3 nodes. Indeed, it is 2-robust.
On the other hand, graph $G_1$ in Figure~\ref{fig:merge-loop}
is 2-robust but not 3-connected (there are 4 v-cliques, where each pair of
adjacent v-cliques share only 1 or 2 nodes).
Accordingly, we can consider a v-union
as a whole in core identification.

Next, we present a necessary condition for $k$-robustness.
\begin{theorem}[$(K+1)$-overlap condition]\label{thm:necessary}
Graph $G$ is $k$-robust only if for every $(k+1)$-connected
v-union $Q\in G$, $Q$ shares at least
$k+1$ common nodes with the subgraph consisting of
the rest of the v-unions. \rbox
\end{theorem}

\begin{proof}
\label{proof:k-overlap}
We prove that if graph $G$ contains a $(k+1)$-connected v-union $Q$ that shares at most $k$ common nodes with the rest of the graph, $G$ is not $k$-robust. Since $Q$ shares at most $k$ common nodes with the subgraph consisting of the rest of the v-unions, removing the common nodes will disconnect $Q$ from $G$, it proves that $G$ is not $k$-robust. Thus, $(k+1)$-overlap condition holds.
\end{proof}

We call a graph $G$ that satisfies the $(k+1)$-overlap condition
a {\em $(k+1)$-overlap graph}. A $k$-robust graph
must be a $(k+1)$-overlap graph but not vice versa.
In Figure~\ref{fig:graph},
subgraph $\{r_{11}-r_{18}\}$ is not a 2-overlap graph, because
there are two 2-connected v-unions,
$\{r_{11}-r_{15}\}$ and $\{r_{13}, r_{16}-r_{18}\}$,
but they share only one node; indeed, the subgraph is not 1-robust.
On the other hand,
graph $G_2$ in Figure~\ref{fig:merge-loop} satisfies the 3-overlap
condition, as it contains four 3-connected v-unions (actually four v-cliques),
$Q_1-Q_4$, and each v-union shares 3 nodes in total with the others;
however, it is not 2-robust (removing $r_3$ and $r_4$ disconnects it).
Accordingly, for $(k+1)$-overlap graphs we still
need to check $k$-robustness by reduction to a Max-flow Problem.

Now the problem is to find $(k+1)$-overlap subgraphs.
Let $G$ be a graph where a $(k+1)$-connected v-union
overlaps with the rest of the v-unions on no more than $k$ nodes.
We split $G$ by removing these overlapping nodes.
For subgraph $\{r_{11}-r_{18}\}$ in Figure~\ref{fig:graph},
we remove $r_{13}$ and
obtain two subgraphs $\{r_{11}-r_{12}, r_{14}-r_{15}\}$ and
$\{r_{16}-r_{18}\}$ (recall from Example~\ref{ex:core}
that $r_{13}$ cannot belong to any core).
Note that the result subgraphs may not
be $(k+1)$-overlap graphs (\eg,
$\{r_{11}-r_{12}, r_{14}-r_{15}\}$ contains two
v-unions that share only one node), so we need to
further screen them.

We now describe our screening algorithm, {\sc Screen} (details in Algorithm~\ref{alg:prune}),
which takes a graph $G$, represented by $\bf C$ and $\bar L$, as input,
finds $(k+1)$-connected
v-unions in $G$ and meanwhile decides if $G$ is a
$(k+1)$-overlap graph. If not, it splits $G$ into subgraphs
for further examination.

\begin{enumerate}\tightlist
  \item If $G$ contains a single node,
  output it as a core if the node represents
  multiple records that belong only to one v-clique.
  \item For each v-clique $C \in \bf C$, initialize a v-union.
  We denote the set of v-unions by $\bar Q$, the v-union that
  $C$ belongs to by $Q(C)$, and the overlapping nodes of $C$ and $C'$
  by $\bar B(C,C')$.
  \item For each v-clique $C \in \bf C$, we merge v-unions
  as follows.

  (a) For each record $r \in C$ that has not been
    considered, for every pair of v-cliques $C_1$ and $C_2$ in $r$'s index entry,
    if they belong to different v-unions, add $r$ to
    overlap $\bar B(C_1,C_2)$.

  (b) For each v-union $Q \ne Q(C)$ where there exist
    $C_1 \in Q$ and $C_2 \in Q(C)$ such that
    $|\bar B(C_1, C_2)|\geq k+1$, merge $Q$ and $Q(C)$.

  At the end, $\bar Q$ contains all $(k+1)$-connected v-unions.
  \item For each v-union $Q \in \bar Q$, find
  its border nodes as $\bar B(Q)=\cup_{C \in Q, C' \not\in Q}\bar B(C,C')$.
  If $|\bar B(Q)|\leq k$, split the subgraph it belongs to,
  denoted by $G(Q)$, into two subgraphs $Q\setminus \bar B(Q)$ and
  $G(Q)\setminus Q$. \eat{If any subgraph contains a single
  node, discard it if it represents a single record,
  and output it as a core if it represents multiple records.}
  \item Return the remaining subgraphs.
\end{enumerate}

{\small
\begin{algorithm}[t]
\caption{{\sc Screening($G, \bar C, \bar L, k$)}\label{alg:prune}
}
\begin{algorithmic}[1]
\REQUIRE $G$: Simplified similarity graph. \\
$\bar C$: Set of $k$-cores.\\
$\bar L$: Inverted list of the similarity graph. \\
$k$: Robustness requirement.

\ENSURE $\bar G$ Set of subgraphs in $G$.

\IF{$G$ contains a single node $r$}

\IF{$r$ represent multiple records}

\STATE add $r$ to $\bar C$.

\ENDIF

\STATE return $\bar G=\phi$.

\ELSE

\STATE initialize v-union $Q(C)$ for each v-clique $C$ and add $Q(C)$ to $\bar Q$.

\STATE // find v-union

\FOR{{\bf each} v-clique $C\in G$}

\FOR{{\bf each} record $r\in C$ that is not proceeded}

\FOR{{\bf each} v-clique pair $C_1, C_2 \in \bar L(r)$}

\IF{$C_1, C_2$ are in different v-unions}

\STATE add $r$ to overlap $\bar B(Q(C_1),Q(C_2))$.

\ENDIF

\ENDFOR

\ENDFOR

\FOR{{\bf each} v-union $Q$ where $\bar B(Q, Q(C))\geq k$}

\STATE merge $Q$ and $Q(C)$ as $Q_m$.

\FOR{{\bf each} v-union $Q'\neq Q, Q'\neq Q(C)$}

\STATE set $\bar B(Q', Q_m)=\bar B(Q',Q)\cup \bar B(Q',Q(C))$

\ENDFOR

\ENDFOR

\ENDFOR

\STATE // screening

\FOR{{\bf each} v-union $Q \in \bar Q$}

\STATE compute $\bar B(Q)=\cup_{Q' \in \bar Q}\bar B(Q,Q')$.

\IF{$|\bar B(Q)|<k$}

\STATE add subgraphs $Q \setminus \bar B(Q)$ and $G(Q)\setminus Q$ into $\bar G$

\ENDIF

\ENDFOR

\ENDIF

\RETURN $\bar G$;
\end{algorithmic}
\end{algorithm}
}

\begin{proposition}\label{prop:screen}
Denote by $|\bar L|$ the number of entries in input $\bar L$.
Let $m$ be the maximum number of values from dominant-value attributes
of a record, and $a$ be the maximum number of adjacent v-unions
that a v-union has. Algorithm {\sc Screen}
finds $(k+1)$-overlap subgraphs in time $O((m^2+a)\cdot|\bar L|)$
and the result is independent of the order
in which we examine the v-cliques.\rbox
\end{proposition}

\begin{proof}
\label{proof:screen}
We first prove the time complexity of {\sc Screen}. It takes in time $O(m^2|\bar L|)$ to scan all entries in $\bar L$ and find common nodes between each pair of adjacent v-cliques (Step 3(a)). It takes in time $O(a|\bf C|)$ to merge v-unions, where $|\bf C|$ is the number of v-cliques in $G$ (Step 3(b)). Since $|\bf C|$ $<|\bar L|$, the algorithm runs in time $O(m^2+a)\cdot |\bar L|$.

We next prove that the result of $\sc Screen$ is independent of the order in which we examine the v-cliques, that is, 1) finding all maximal $(k+1)$-connected v-unions in $G$ is order independent; 2) removing all nodes in $\bar B(Q)$ from $G$ where $|\bar B(Q)|\leq k$ is order independent.

Consider order independency of finding all v-unions in $G$. To find all v-unions in $G$ is conceptually equivalent to find all connected components in an abstract graph $G_A$, where each node in $G_A$ is a v-clique in $G$ and two nodes in $G_A$ are connected if the two corresponding v-cliques share more than $k$ nodes. {\sc Screen} checks whether each node in $G$ is a common node between two v-cliques (Step 3(a)), and if two cliques share more than $k$ nodes, merges their v-unions (Step 3(b)), which is equivalent to connect two nodes in $G_A$. Once all nodes in $G$ is scanned, all edges in $G_A$ are added, and the order in which we examine nodes in $G$ is independent from the structure of $G_A$ and the connected components in $G_A$. Therefore, finding all v-unions in $G$ is order independent.

Consider order independency of removing nodes in $G$. Suppose $Q_1, Q_2, ..., Q_m, m>0$ are all v-unions in $G$ with $|\bar B(Q_i)|\leq k, i\in[1, m]$. Since $G$ is finite, $Q_i$ is finite and unique; thus, removing all nodes in $\bar B(Q))$ from $G$ where $|\bar B(Q)|\leq k$ is order independent.
\end{proof}

\eat{
Note first that each result subgraph is not necessarily a
$k$-overlap subgraph, as we have just illustrated,
so we need to apply this algorithm iteratively.
Note also that a result $k$-overlap subgraph is not necessarily
$(k-1)$-robust, since $k$-overlap condition is only a necessary
condition, so we still need to check its $k$-robustness.
In that checking we can treat a v-union as a whole. }

Note that $m$ and $a$ are typically very small, so {\sc Screen}
is basically linear in the size of the inverted index.
Finally, we have results similar to Theorem~\ref{thm:simplify}
for v-unions, so we can further simplify the graph
by keeping for each v-union a single representative
for all nodes that only belong to it.
Each result $k$-overlap subgraph is typically very small.

\begin{example}
Consider Table~\ref{tbl:inverted} as input and $k=1$.
Step 2 creates five v-unions $Q_1-Q_5$ for the five
v-cliques in the input.

Step 3 starts with v-clique $C_1$. It has 4 nodes (in the simplified
inverted index), among which 3 are shared with $C_2$.
Thus, $\bar B(C_1, C_2)=\{r_{3}-r_{5}\}$ and
$|\bar B(C_1,C_2)| \geq 2$, so we merge $Q_1$ and $Q_2$ into $Q_{1/2}$.
Examining $C_2$ reveals no other shared node.

Step 3 then considers v-clique $C_3$. It has three nodes,
among which $r_{12}-r_{13}$ are shared with $C_4$ and $r_{13}$
is also shared with $C_5$. Thus, $\bar B(C_3, C_4)=\{r_{12}-r_{13}\}$
and $\bar B(C_3, C_5)=\{r_{13}\}$. We merge $Q_3$ and $Q_4$
into $Q_{3/4}$. Examining $C_4$ and $C_5$ reveals no other shared
node. We thus obtain three 2-connected v-unions:
$\bar Q=\{Q_{1/2}, Q_{3/4}, Q_5\}$.

Step 4 then considers each v-union. For $Q_{1/2}$,
$\bar B(Q_{1/2})=\emptyset$ and we thus split subgraph $Q_{1/2}$ out
and merge all of its nodes to one $r_{1/\dots/7}$.
For $Q_{3/4}$, $\bar B(Q_{3/4})=\{r_{13}\}$ so $|\bar B(Q_{3/4})|<2$.
We split $Q_{3/4}$ out and obtain $\{r_{11}-r_{12}, r_{14/15}\}$
($r_{13}$ is excluded).
Similar for $Q_5$ and we obtain $\{r_{16/17/18}\}$.
Therefore, we return three subgraphs for further screening. \rbox
\end{example}

\subsubsection{Reduction}\label{subsubsec:split}
Intuitively, a graph $G(V,E)$ is $k$-robust if and only if between any two nodes
$a,b \in V$, there are more than $k$ paths that do not share any node
except $a$ and $b$. We denote the number of non-overlapping
paths between nodes $a$ and $b$ by $\kappa(a,b)$. We can reduce
the problem of computing $\kappa(a,b)$ into a Max-flow Problem.

For each input $G(V,E)$ and nodes $a,b$, we construct
the (directed) {\em flow network} $G'(V',E')$ as follows.
\begin{enumerate}\tightlist
  \item Node $a$ is the source and $b$ is the sink
      (there is no particular order between $a$ and $b$).
  \item For each $v \in V, v \ne a, v \ne b$, add two nodes
  $v',v''$ to $V'$, and two directed edges $(v',v''), (v'',v')$ to $E'$.
  If $v'$ represents $n$ nodes, the edge $(v',v'')$ has
  weight $n$, and the edge $(v'',v')$ has weight $\infty$.
\item For each edge $(a,v) \in E$, add edge $(a,v')$ to $E'$;
      for each edge $(u,b) \in E$, add edge $(u'',b)$ to $E'$;
      for each other edge $(u,v) \in E$, add two edges $(u'',v')$
      and $(v'',u')$ to $E'$. Each edge has capacity $\infty$.
\end{enumerate}

\begin{lemma}
The max flow from source $a$ to sink $b$ in $G'(V',E')$ is equivalent
to $\kappa(a,b)$ in $G(V,E)$. \rbox
\end{lemma}
\begin{proof}
\label{proof:max-flow}
According to Menger's Theorem~\cite{Bruhn:2005:MTI:1384589.1384594}, the minimum number of nodes whose removal disconnects $a$ and $b$, that is $\kappa(a,b)$, is equal to the maximum number of independent paths between $a$ and $b$.  The authors in~\cite{citeulike:2529515} proves that the maximum number of independent paths between $a$ and $b$ in an undirected graph $G(V,E)$ is equivalent to the maximal value of flow from $a$ to $b$ or the minimal capacity of an $a-b$ cut, the set of nodes such that any path from $a$ to $b$ contains a member of the cut, in $G'(V',E')$.
\end{proof}
\begin{example}\label{ex:maxflow}
Consider nodes $r_{1}$ and $r_6$ of graph $G_2$ in
Figure~\ref{fig:merge-loop}.
Figure~\ref{fig:maxflow} shows the corresponding flow network,
where the dash line (across edges $(r_3',r_3''), (r_4',r_4'')$)
in the figure cuts the flow from $r_1$ to
$r_6$ with a minimum cost of 2.
The max flow/min cut has value 2. Indeed, $\kappa(r_{1},r_6)=2$.\rbox
\end{example}

\begin{figure}[t]
\centering
\vspace{-.2in}
\includegraphics[scale=.35]{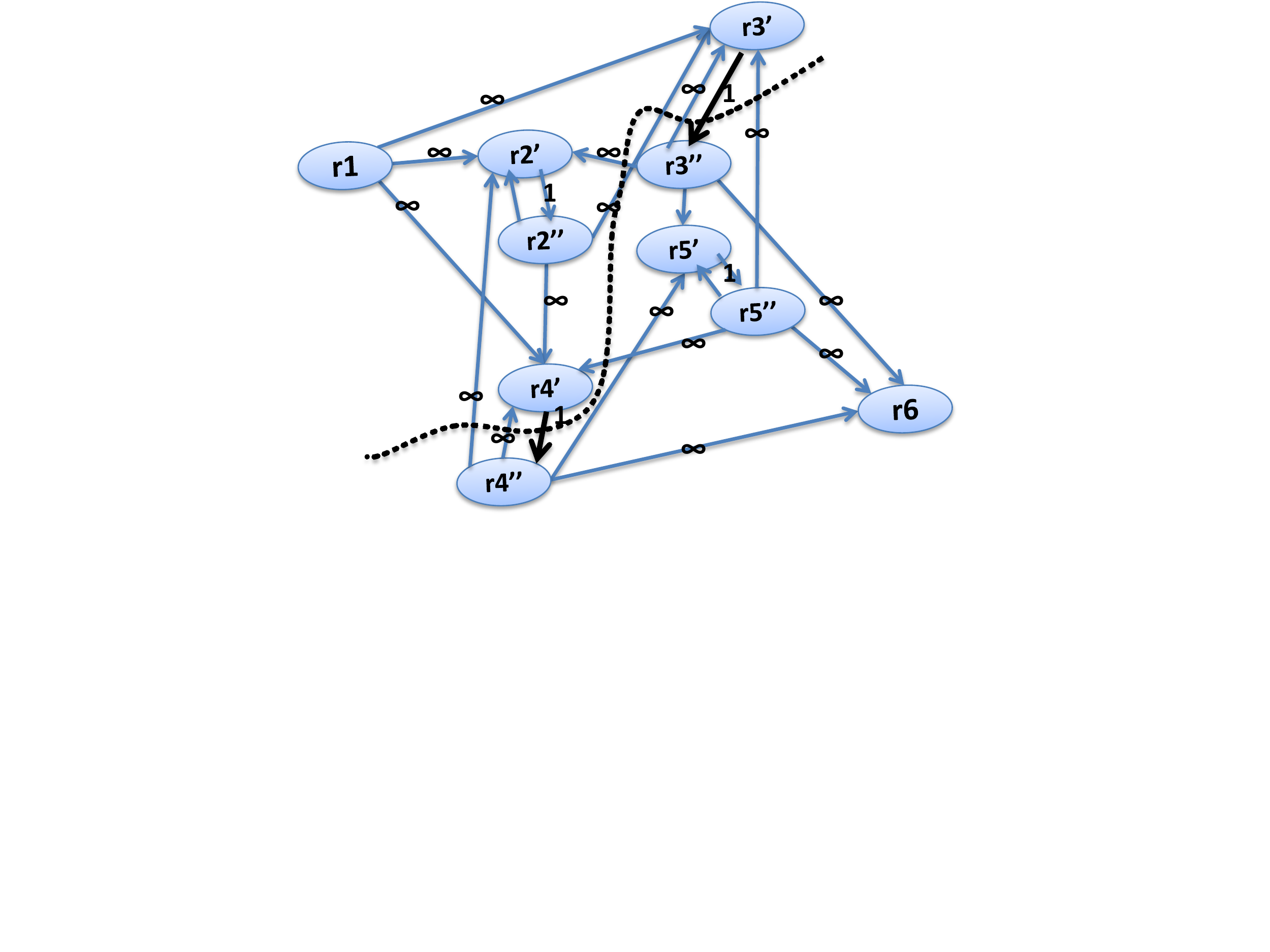}
\vspace{-1.4in} {\small\caption{Flow network for $G_2$
in Figure~\ref{fig:merge-loop}.\label{fig:maxflow}}}
\vspace{-.15in}
\end{figure}

Recall that in a $(k+1)$-connected v-union, between each pair of nodes
there are at least $k+1$ paths. Thus, if (1) $\kappa(a,b)=k+1$,
(2) $a$ and $b$ belong to different v-unions,
and (3) $a$ and $a'$ belong to the same v-union, we must have
$\kappa(a',b)\geq k+1$. We thus have the following sufficient and necessary
condition for $k$-robustness.
\begin{theorem}[Max-flow condition]
\label{thm:maxflow}
Let $G(V,E)$ be an input similarity graph. Graph $G$ is $k$-robust
if and only if for every pair of adjacent $(k+1)$-connected
v-unions $Q$ and $Q'$, there exist two nodes
$a \in Q\setminus Q'$ and $b \in Q'\setminus Q$
such that the max flow from $a$ to $b$ in the corresponding
flow network is at least $k+1$.\rbox
\end{theorem}
\begin{proof}
\label{proof:max-flow}
According to Menger's Theorem~\cite{Bruhn:2005:MTI:1384589.1384594}, $\kappa(a,b)$ in $G$ is equivalent to the max-flow from $a$ to $b$ in the corresponding flow network. We need to prove that graph $G$ is $k$-robust if and only if for each pair of adjacent $(k+1)$-connected v-unions $Q$ and $Q'$, there exists two nodes $a \in Q\setminus Q'$ and $b \in Q'\setminus Q$ such that $\kappa(a, b)\geq k+1$.

We first prove that if $G$ is $k$-robust, for each pair of adjacent $(k+1)$-connected v-unions $Q$ and $Q'$, there exists two nodes $a \in Q\setminus Q'$ and $b \in Q'\setminus Q$ such that $\kappa(a, b)\geq k+1$. Since $G$ is $k$-robust, for each pair of nodes $a$ and $b$ in $G$, we have $\kappa(a,b)\geq k+1$.

We next prove that if $G$ is not $k$-robust, there exists a pair of adjacent $(k+1)$-connected v-unions $Q$ and $Q'$ such that for each pair of nodes $a \in Q\setminus Q'$ and $b \in Q'\setminus Q$, we have $\kappa(a, b)< k+1$. Since $G$ is not $k$-robust, there exists a separator $\bar S$, a set of nodes in $G$ with size no greater than $k$ whose removal disconnects $G$ into two sub-graphs $\bar X$ and $\bar Y$. Suppose $Q$ and $Q'$ are two v-unions in $G$ such that $Q\subseteq \bar X, Q'\subseteq \bar Y$ and $Q\cap Q' \neq \emptyset$. For each pair of nodes $a \in Q\setminus Q'$ and $b \in Q'\setminus Q$, we have $a\in \bar X$ and $b\in \bar Y$, and removing the set of nodes in $\bar S$ disconnects $a$ and $b$; thus $\kappa(a,b)<k+1$.

The above two cases proves that graph $G$ is $k$-robust if and only if for every pair of adjacent $(k+1)$-connected v-unions $Q$ and $Q'$, there exist two nodes $a \in Q\setminus Q'$ and $b \in Q'\setminus Q$ such that $\kappa(a, b)\geq k+1$, i.e. the max flow from $a$ to $b$ in the corresponding flow network is at least $k+1$.
\end{proof}

If a graph $G$ is not $k$-robust, we shall split it into subgraphs
for further processing.
In the corresponding flow network, each edge in the minimum cut
must be between a pair of nodes derived from the same node
in $G$ (other edges have capacity $\infty$).
These nodes cannot belong to any core and we use them
as {\em separator} nodes, denoted by $\bar S$. Suppose the separator
separates $G$ into $\bar X$ and $\bar Y$
(there can be more subgraphs); we return
$\bar X \cup \bar S$ and $\bar Y \cup \bar S$.

Note that we need to include $\bar S$ in both sub-graphs
to maintain the integrity of each v-union.
To understand why, consider $G_2$ in Figure~\ref{fig:merge-loop}
where $\bar S=\{r_3,r_4\}$. According to the definition,
there is no 2-core.
If we split $G_2$ into $\{r_1-r_2\}$ and $\{r_5-r_6\}$
(without including $\bar S$),
both subgraphs are 2-robust and we would return them as 2-cores.
The problem happens because v-cliques $Q_1-Q_4$ ``disappear''
after we remove the separators $r_3$ and $r_4$.
Thus, we should split $G_2$ into $\{r_{1}-r_4\}$ and
$\{r_3-r_{6}\}$ instead
and that would further trigger splitting on both subgraphs.
Eventually we wish to exclude the separator nodes
from any core, so we mark them as ``separators'' and
exclude them from the returned cores.

Algorithm {\sc Split} (details in Algorithm~\ref{alg:split})
takes a $(k+1)$-overlap subgraph $G$ as input and
decides if $G$ is $k$-robust. If not, it splits $G$
into subgraphs on which we will then re-apply screening.
\begin{enumerate}\tightlist
  \item For each pair of adjacent $(k+1)$-connected
    v-unions $Q, Q' \in G$,
    find $a \in Q\setminus Q', b \in Q'\setminus Q$. Construct
    flow network $G'(V',E')$ and apply Ford \& Fulkerson
    Algorithm~\cite{F.F.62} to compute the max flow.
  \item Once we find nodes $a,b$ where $\kappa(a,b)\leq k$,
    use the min cut of the flow network as separator
    $\bar S$. Remove $\bar S$ and obtain several subgraphs.
    Add $\bar S$ back to each subgraph and
    mark $\bar S$ as ``separator''. Return the subgraphs for screening.
  \item Otherwise, $G$ is $k$-robust and output it as a $k$-core.
\end{enumerate}

{\small
\begin{algorithm}[t]
\caption{{\sc Split($G, \bar C, k$)}\label{alg:split}
}
\begin{algorithmic}[1]
\REQUIRE $G$: Simplified similarity graph. \\
$\bar C$: Set of cores.\\
$k$: Robustness requirement.

\ENSURE $\bar G$ Set of subgraphs in $G$.

\FOR{{\bf each} adjacent $(k+1)$-connected v-unions $Q, Q'$}

\STATE find a pair of nodes $a\in Q\setminus Q', b \in Q'\setminus Q$.

\STATE construct flow-network $G'$ and compute $\kappa(a,b)$ by Ford \& Fulkerson Algorithm.

\IF{$\kappa(a,b)\leq k$}

\STATE get separator $\bar S$ from $G'$ and remove $\bar S$ from $G$ to obtain disconnected subgraphs;
mark $\bar S$ as ``separator'' and add it to each subgraph in $G$.

\STATE return the set $\bar G$ of subgraphs.

\ENDIF

\ENDFOR

\IF{$\bar G=\phi$}

\STATE add $G$ to $\bar C$.

\ENDIF

\RETURN $\bar G$;
\end{algorithmic}
\end{algorithm}
}

\begin{example}\label{ex:split}
Continue with Example~\ref{ex:maxflow} and $k=2$. 
There are four 3-connected v-unions.
When we check $r_1 \in Q_1$ and $r_6 \in Q_3$, we find 
 $\bar S=\{r_3, r_4\}$. We then split $G_2$ into subgraphs
$\{r_{1}-r_4\}$ and $\{r_3-r_{6}\}$, marking $r_3$ and $r_4$
as ``separators''.

Now consider graph $G_1$ in Figure~\ref{fig:merge-loop}
and $k=2$. 
There are four 3-connected v-unions (actually four v-cliques) and
six pairs of adjacent v-unions. For $Q_1$ and $Q_2$,
we check nodes $r_2$ and $r_4$ and find $\kappa(r_2,r_4)=3$.
Similarly we check for every other pair of adjacent
v-unions and decide that the graph is 2-robust. \rbox
\end{example}

\begin{proposition}\label{prop:split}
Let $p$ be the total number of pairs of adjacent v-unions,
and $g$ be the number of nodes in the input graph.
Algorithm {\sc Split} runs in time $O(pg^{2.5})$. \rbox
\end{proposition}

\begin{proof}
\label{proof:split}
Authors in~\cite{citeulike:2529515} proves that it takes in time $O(g^{2.5})$ to compute $\kappa(a, b)$ for a pair of nodes $a$ and $b$ in $G$. In the worst case {\sc Split} needs to compute $\kappa(a,b)$ for $p$ pairs of adjacent v-unions. Thus, {\sc Split} runs in time $O(pg^{2.5})$.
\end{proof}

Recall that if we solve the Max-Flow Problem directly for
each pair of sources in the original graph, the complexity
is $O(|\bar L|^{4.5})$, which would be dramatically higher.

\subsubsection{Full algorithm}\label{subsubsec:core-detection}
We are now ready to present the full algorithm,
{\sc Core} (Algorithm~\ref{alg:core}).
Initially, it initializes the
working queue $\bf Q$ with only input $G$ (Line~\ref{ln:init}).
Each time it pops a subgraph $G'$ from $\bf Q$ and invokes
{\sc Screen} (Lines~\ref{ln:pop}-\ref{ln:screen}).
If the output of {\sc Screen} is still $G'$
(so $G'$ is a $(k+1)$-overlap subgraph) (Line~\ref{ln:same}),
it removes any node with mark ``separator'' in $G'$ and puts the
new subgraph into the working queue (Line~\ref{ln:sep}),
or invokes {\sc Split} on $G'$ if there is no separator
(Line~\ref{ln:split}). Subgraphs output by {\sc Screen} and
{\sc Split} are added to the queue for further examination
(Lines~\ref{ln:queue1}, \ref{ln:queue2}) and identified cores
are added to $\bar C$, the core set. It terminates
when $\bf Q=\emptyset$.

{\small
\begin{algorithm}[t]
\caption{{\sc Core($G, k$)}\label{alg:core}
}
\begin{algorithmic}[1]
\REQUIRE $G$: Simplified similarity graph, represented by $\bf C$ and $\bar L$. \\
$k$: Robustness requirement.

\ENSURE $\bar C$ Set of cores in $G$.

\STATE Let ${\bf Q}=\{G\}$, $\bar C=\emptyset$; \label{ln:init}

\WHILE{$\bf Q\neq \phi$}

\STATE Pop $G'$ from $\bf Q$;\label{ln:pop}

\STATE Let $\bar P=$ {\sc Screen($G', k, \bar C$)};\label{ln:screen}

\IF{$\bar P=\{G'\}$}\label{ln:same}

\IF{$G'$ contains ``separator'' nodes}

\STATE Remove separators from $G'$ and add the result to
       $\bf Q$ if it is not empty;\label{ln:sep}

\ELSE

\STATE Let $\bar S=$ {\sc Split}($G', k, \bar C$);\label{ln:split}

\STATE add graphs in $\bar S$ to $\bf Q$;\label{ln:queue1}

\ENDIF

\ELSE

\STATE add graphs in $\bar P$ to $\bf Q$;\label{ln:queue2}

\ENDIF

\ENDWHILE

\RETURN $\bar C$;
\end{algorithmic}
\end{algorithm}
}

The correctness of algorithm {\sc Core} is guaranteed by the following Lemmas.

\begin{lemma}\label{lem:core-parition}
For each pair of adjacent nodes $r, r'$ in graph $G$, there exists a maximal $k$-robust partitioning such that $r, r'$ are in the same subgraph.\rbox
\end{lemma}
\begin{proof}
\label{proof:core-partition}
For each pair of adjacent nodes $r, r'$ in $G$, we prove the existence of such a maximal $k$-robust partitioning by constructing it.

By definition, adjacent node $r, r'$ form a v-clique $C$. Therefore, there exists a maximal v-clique $C'$ in $G$ that contains $r, r'$, i.e., $C \subseteq C'$. V-clique $C'$ can be obtained by keep adding nodes in $G$ to $C$ so that each newly-added node is adjacent to each node in current clique until no nodes in $G$ can be added to $C'$. By definition, any v-clique is $k$-robust, therefore there exists a maximal $k$-robust sub-graph $G'$ in $G$ such that $C' \subseteq G'$. Graph $G'$ can be obtained by keep adding nodes in $G$ to $C'$ so that each newly-added node is adjacent to at least $k+1$ nodes in current graph $G'$ until no nodes in $G$ can be added to $G'$. We remove $G'$ from $G$ and take $G'$ as a subgraph in the desired partitioning.

We repeat the above process to a randomly-selected pair of adjacent nodes in the remaining graph $G\setminus G'$ until it is empty. The desired partitioning satisfies Condition 1 and 2 of Definition~\ref{def:core-clustering} because the above process makes sure each subgraph is exclusive and $k$-robust; it satisfies Condition 3 of Definition~\ref{def:core-clustering} because the above process makes sure each subgraph is maximal, which means merging arbitrary number of subgraphs in the partitioning would violate Condition 2.

In summary, the desired partitioning is a maximal $k$-robust partitioning. It proves that for each pair of adjacent nodes $r$ and $r'$ in graph $G$, there exists a maximal $k$-robust partitioning such that $r$ and $r'$ are in the same subgraph.
\end{proof}
\begin{lemma}\label{lem:separator}
The set of nodes in a separator $\bar S$ of graph $G$ does not belong to any $k$-core in $G$, where $|\bar S|\leq k$.\rbox
\end{lemma}
\begin{proof}[Lemma~\ref{lem:separator}]\label{proof:separator}
Suppose the set $\bar S$ of nodes separate $G$ into $m$ disconnected sets $\bar X_i, i\in[1, m], m>0$. To prove that each node $r\in \bar S$ does not belong to any $k$-core in $G$, we prove that for a node $r' \in G, r'\neq r$, there exists a maximal $k$-robust partitioning such that $r$ and $r'$ are separated. Node $r'$ falls into the following cases: 1) $r' \in \bar X_i, i\in[1, m]$ ; 2) $r' \in \bar S$.

Consider Case 1) where $r' \in \bar X_i, i\in[1, m]$. We construct a maximal $k$-robust partitioning of $G$ where $r$ and $r'$ are in different subgraphs. We start with a maximal $k$-robust subgraph $G'$ in $G$ that contains $r$ and $r''$ where $r''$ is adjacent to $r$ and in $\bar X_j, j\neq i, j\in[1, m]$, and find other maximal $k$-robust subgraphs as in Lemma~\ref{lem:core-parition}. Since $\bar S$ separates $\bar X_i$ and $\bar X_j$, maximal $k$-robust subgraph $G'$ that contains $r$ and $r''$ does not contain any node in $\bar X_i$. It proves that there exists a maximal $k$-robust partitioning of $G$ where $r$ and $r'$ are not in the same subgraph.

Consider Case 2) where $r'\in \bar S$. We construct a maximal $k$-robust partitioning of $G$ such that $r$ and $r'$ are in different subgraphs. We create two maximal $k$-robust subgraphs $G'$ and $G''$, where $G'$ contains $r$ and an adjacent node $r_i \in \bar X_i, i\in[1, m]$, $G''$ contains $r'$ and an adjacent node $r_j \in \bar X_j, j\neq i, j\in [1, m]$. We create other subgraphs as in Lemma~\ref{lem:core-parition}. Since each path between $r_i \in \bar X_i$ and $r_j \in \bar X_j$ contains at least one node in $\bar S$ and $|\bar S|\leq k$, graph $G' \cup G''$ is not $k$-robust. Therefore, the created partitioning is a maximal $k$-robust partitioning. It proves that there exists a maximal $k$-robust partitioning of $G$ where $r$ and $r'$ are not in the same subgraph.

Given the above two cases, we have that any node in separator $\bar S$ of $G$ does not belong to any $k$-core in $G$, where $|\bar S|\leq k$.
\end{proof}

\begin{theorem}\label{pro:core}
Let $G$ be the input graph and $q$ be the number of $(k+1)$-connected
v-unions in $G$. Define $a,p,g,m,$ and $|\bar L|$
as in Proposition~\ref{prop:screen} and \ref{prop:split}.
Algorithm {\sc Core} finds correct $k$-cores of $G$ in
time $O(q((m^2+a)|\bar L|+pg^{2.5}))$ and is order independent.\rbox
\end{theorem}

\begin{proof}
\label{proof:core}
We first prove that {\sc Core} correctly finds $k$-cores in $G$, that is 1) nodes not returned by {\sc Core} do not belong to any $k$-core; 2) each subgraph returned by {\sc Core} forms a $k$-core.

We prove that nodes not returned by {\sc Core} do not belong to any $k$-core in $G$. Nodes not returned by {\sc Core} belong to separators of subgraphs in $G$. Suppose $\bar S$ is a separator of graph $G_n\in \bf Q$ found in either {\sc Screen} or {\sc Split} phase, where $G_n\subseteq G, n\geq 0, G_0=G$, and $\bar S$ separates $G_n$ into $m$ sub-graphs $\bar X_n^i, i\in [1,m], m>1$. Graph $G_n^i \in \bf Q$ is a subgraph of $G_n$ such that any node $r\in \bar X_n^j, j\in[1, m], j\neq i$ does not belong to $G_n^i$. Nodes removed in $G_n^i$ by {\sc Core} belong to separator $\bar S$ in $G_n$. Given Lema~\ref{lem:separator}, such nodes do not belong to any $k$-core in $G_n$ and thus does not belong to any $k$-core in $G$.

We next prove that each subgraph returned by {\sc Core} forms a $k$-core in $G$. We prove two cases: 1) subgraph $G'$ in $G$ forms a $k$-core if there exists a separator $\bar S$ that disconnects $G'$ from $G$, where $|\bar S|\leq k$ and $G' \cup \bar S$ and $G'$ are both $k$-robust; 2) if a subgraph is a $k$-core in $G_n^i$, it is a $k$-core in graph $G_n$.

We consider Case 1) that subgraph $G'$ in $G$ forms a $k$-core if there exists a separator $\bar S$ that disconnects $G'$ from $G$, where $|\bar S|\leq k$ and $G' \cup \bar S$ and $G'$ are both $k$-robust. For a pair of nodes $r_1, r_2$ in $G'$, we prove that there exists no maximal $k$-robust partitioning where $r_1$ and $r_2$ are in different subgraphs. Suppose such a partitioning exists, and $G_1, G_2$ are subgraphs containing $r_1, r_2$ respectively. Since $G_1, G_2 \subseteq G'\cup \bar S$, we have that $G_1\cup G_2$ is $k$-robust, it violates the fact that the result of merging any two subgraphs in a maximal $k$-robust partitioning is not $k$-robust. Therefore, there exists no maximal $k$-robust partitioning where $r_1$ and $r_2$ are in different subgraphs. It proves that $G'$ is a $k$-core in $G$.

We next consider Case 2) that if a subgraph $G'$ is a $k$-core in $G_n^i$, it is a $k$-core in graph $G_n$. 
We prove that a pair of nodes $r_1, r_2\in G'$ belong to the same subgraph of all maximal $k$-robust partitioning in $G_n$. Suppose there exists such a partitioning of $G_n$ where $r_1\in G_1, r_2\in G_2$. Since $G_n^i \subseteq \bar X_n^i \cup \bar S$, we have $G_1, G_2\subseteq G_n^i$, otherwise $G_1, G_2$ are not $k$-robust. Since $r_1, r_2$ belong to the same $k$-core in $G_n^i$, we have $G_1=G_2$. It proves that if $G'$ is a $k$-core in $G_n^i$, it is a $k$-core in $G_n$.

The above two cases prove that each subgraph returned by {\sc Core} forms a $k$-core in $G$. In summary, nodes not returned by {\sc Core} do not belong to any $k$-core, and each subgraph returned by {\sc Core} forms a $k$-core in $G$. Thus, {\sc Core} correctly finds all $k$-cores in $G$. It further proves that the result of {\sc Core} is independent from the order in which we find and remove separators of graphs in $\bf Q$.

We now analyze the time complexity of {\sc Core}. For each $(k+1)$-connected v-unions in $G$, it takes in time $O(m^2+a)|\bar L|$ to proceed {\sc Screen} phase and in time $O(pg^{2.5})$ to proceed {\sc Split} phase. In total there are $q$ v-unions in $G$, thus the algorithm takes in time $O(q((m^2+a)|\bar L|+pg^{2.5}))$.
\end{proof}

\begin{table}
\scriptsize
\vspace{-.1in}
{\small\caption{\label{tbl:core}Step-by-step core identification
in Example~\ref{ex:detect}.}}
\begin{center}
{\small
\begin{tabular}{|c|c|l|}
\hline
Input & Method & \multicolumn{1}{c|}{Output} \\
\hline
\hline
$G_2$ & {\sc Screen} & $G_2$ \\
\hline
$G_2$ & {\sc Split} & $G_2^1=\{r_{1}-r_4\}, G_2^2=\{r_3-r_{6}\}$ \\
\hline
$G_2^1$ & {\sc Screen} &$G_2^3=\{r_3\}, G_2^4=\{r_4\}$ \\
\hline
$G_2^2$ & {\sc Screen} &$G_2^3=\{r_3\}, G_2^4=\{r_4\}$ \\
\hline
$G_2^3$ & {\sc Screen} & -\\
\hline
$G_2^4$ & {\sc Screen} & -\\
\hline
\hline
$G$ & {\sc Screen} & $G^1=\{r_{1/\dots/7}\}, G^2=\{r_{11}, r_{12}, r_{14/15}\},$ \\
&  & $G^3=\{r_{16/17/18}\}$ \\
\hline
$G^1$ & {\sc Screen} & Core $\{r_{1}-r_{7}\}$ \\
\hline
$G^2$ & {\sc Screen} & $G^4=\{r_{11}\}, G^5=\{r_{14/15}\}$ \\
\hline
$G^3$ & {\sc Screen} & Core $\{r_{16}-r_{18}\}$ \\
\hline
$G^4$ & {\sc Screen} & -\\
\hline
$G^5$ & {\sc Screen} & Core $\{r_{14}-r_{15}\}$ \\
\hline
\end{tabular}
}
\end{center}
\vspace{-.25in}
\end{table}
\vspace{-.1in}
\begin{example}\label{ex:detect}
First, consider graph $G_2$ in Figure~\ref{fig:merge-loop} and $k=2$.
Table~\ref{tbl:core}
shows the step-by-step core identification process.
It passes screening and is the input for {\sc split}.
{\sc Split} then splits it into $G_2^1$ and $G_2^2$,
where $r_3$ and $r_4$ are marked as ``separators''.
{\sc Screen} further splits each of them
into $\{r_3\}$ and $\{r_4\}$, both discarded as each
represents a single node (and is a separator).
So {\sc Core} does not output any core.

Next, consider the motivating example, with the input
shown in Table~\ref{tbl:inverted} and $k=1$.
Originally, ${\bf Q}=\{G\}$. After invoking {\sc Screen}
on $G$, we obtain three subgraphs $G^1, G^2,$ and $G^3$.
{\sc Screen} outputs $G^1$ and $G^3$ as 1-cores
since each contains a single node that represents
multiple records. It further splits $G^2$
into two single-node graphs $G^4$ and $G^5$,
and outputs the latter as a 1-core. Note that if we remove
the 1-robustness requirement, we would merge $r_{11}-r_{18}$
to the same core and get false positives.\rbox
\end{example}

\smallskip
\noindent
{\bf Case study:} On the data set with 18M records,
our core-identification algorithm finished in 2.2 minutes.
{\sc Screen} was invoked 114K times
and took 2 minutes (91\%) in total.
Except the original graph, an input contains 
at most 39.3K nodes; for 97\% inputs there are fewer than 10 nodes
and running {\sc Screen} was very fast.
{\sc Split} was invoked only 26 times;
an input contains at most 65 nodes (13 v-unions) and on average 7.8 (2.7 v-unions).
Recall that the simplified inverted index contains 1.5M entries,
so {\sc Screen} reduced the size of the input to {\sc Split}
by 4 orders of magnitude.

\section{Group Linkage}\label{sec:sate}
The second stage clusters the cores and the remaining records,
which we call {\em satellites}, into groups.
To avoid merging records based only on weak evidence,
we require that {\em a cluster cannot contain more than one
satellite but no core}. Comparing with clustering in traditional
record linkage, our algorithm differs in three aspects.
First, in addition to weighting each attribute, we weight the values according
to their popularity within a group such that similarity
on primary values (strong evidence) is rewarded more.
Second, we treat all values for dominant-value attributes
as a whole, we are tolerant to differences on local values
from different entities in the same group. Third, we distinguish
weights for distinct values and non-distinct values
such that similarity on distinct values is rewarded more.
This section first describes the objective
function for clustering (Section~\ref{sec:obj}) and then
proposes a greedy algorithm for clustering (Section~\ref{sec:algo}).


\subsection{Objective function}\label{sec:obj}
\noindent
{\bf SV-index:}
Ideally, we wish that each cluster is {\em cohesive} (each element,
being a core or a satellite,
is close to other elements in the same cluster) and different
clusters are {\em distinct} (each element is fairly different
from those in other clusters). Since records
in the same group may have fairly different local values,
we adopt {\em Silhouette Validation
Index (SV-index)}~\cite{SV-index} as the objective function as it is more tolerant
to diversity within a cluster. Given a clustering $\cal C$ of
elements $\bf E$, the SV-index of $\cal C$ is defined as follows.




\vspace{-.15in} {\small
\begin{eqnarray}
\label{eqn:sv}
S(\cal C)&=& Avg_{e \in {\bf E}} S(e);\\
S(e) &=& \frac{a(e)-b(e)+\alpha}{\max\{a(e), b(e)\}+\beta}.
\end{eqnarray}
} \vspace{-.1in}

\smallskip
\noindent
Here, $a(e) \in [0,1]$ denotes the similarity between element $e$ and its own cluster, $b(e) \in [0,1]$ denotes the maximum similarity between $e$ and another cluster, $\beta > \alpha > 0$ are small numbers to keep $S(e)$ finite and non-zero (we discuss in Section~\ref{sec:experiment} how we set the parameters).
A nice property of $S(e)$ is that it falls in $[-1, 1]$, where a value close to $1$ indicates that $e$ is in an appropriate cluster, a value close to $-1$ indicates that $e$ is mis-classified, and a value close to $0$ while $a(e)$ is not too small indicates that $e$ is equally similar to two clusters that should possibly be merged. Accordingly, we wish to obtain a clustering with the maximum SV-index. We next describe how we compare an element with a cluster.

\smallskip
\noindent
{\bf Similarity computation:} We consider that an element $e$ is similar to a cluster $Cl$ if they have highly similar values on common-value attributes (\eg, {\sf name}), share at least one {\em primary} value (we explain ``primary'' later) on dominant-value attributes (\eg, {\sf phone, URL}); in addition, our confidence is higher if they also share values on
multi-value attributes (\eg, {\sf category}).
Following previous work on handling multi-value
attributes~\cite{recon, LDMS11},
we compute the similarity $sim(e, Cl)$ as follows.

\vspace{-.15in} {\small
\begin{eqnarray}
\label{eqn:sim}
sim(e,Cl) &=& \min\{1, sim_s(e,Cl)+\tau w_msim_{multi}(e,Cl)\}; \\
sim_s(e,Cl) &=& \frac{w_{c}sim_{com}(e, Cl)+w_osim_{dom}(e, Cl)}{w_c+w_o};\\
\tau &=& \left\{
\begin{array}{rl}
0 & \text{if}\ sim_s(e,Cl) < \theta_{th},\\
1 & \text{otherwise}.
\end{array}\right.
\end{eqnarray}
} \vspace{-.1in}

\smallskip
\noindent Here, $sim_{com}, sim_{dom},$ and $sim_{multi}$ denote the similarity for common-, dominant-, and multi-attributes respectively.
We take the weighted sum of $sim_{com}$ and $sim_{dom}$ as strong indicator of $e$ belonging to $Cl$ (measured by $sim_s(e,Cl)$), and only reward weak indicator $sim_{multi}$ if $sim_s(e,Cl)$ is above a pre-defined threshold $\theta_{th}$; the similarity is at most 1. Weights $0<w_c, w_o, w_m<1$ indicate how much we reward value similarity or penalize value difference; we learn the weights from sampled data. We next highlight how we
leverage strong evidence from cores and
meanwhile remain tolerant to other different values in similarity computation.

First, we identify {\em primary values} (strong evidence)
as popular values within a cluster.
When we maintain the signature for a core or a cluster,
we keep all values of an attribute and assign a high {\em weight}
to a popular value. Specifically, let $\bar R$ be a set of records.
Consider value $v$ and let $\bar R(v) \subseteq \bar R$ denote the records
in $\bar R$ that contain $v$. The weight of $v$ is computed by
$w(v)={|\bar R(v)| \over |\bar R|}$.

\begin{example}
\label{ex:weight}
Consider {\sf phone} for core $\mbox{Cr}_1=\{r_{1}-r_{7}\}$ in Table~\ref{tab:motiv}. There are $7$ business listings in $\mbox{Cr}_1$, $5$ providing {\em 808} ($r_{1}-r_{5}$), one providing {\em 101} ($r_{6}$), and one providing {\em 102} ($r_{7}$).
Thus, the weight of {\em 808} is $\frac{5}{7}=.71$ and the weight for {\em 101} and {\em 102} is ${1 \over 7} = .14$, showing that {\em 808} is the primary
phone for $\mbox{Cr}_1$. \rbox
\end{example}

Second, when we compute $sim_{dom}(e,Cl)$, we consider all the dominant-value
attributes together, rewarding sharing primary values (values with a high
weight) but not penalizing different values unless there is no shared value.
Specifically, if the primary value of an element is the same as that of a cluster,
we consider them having probability $p$ to be in the same group. Since we use weights to measure whether the value is primary and allow slight difference on values, with a value $v$ from $e$ and $v'$ from $Cl$, the probability becomes $p\cdot w_e(v) \cdot w_{Cl}(v') \cdot s(v, v')$, where $w_e(v)$ measures the weight of $v$ in $e$, $w_{Cl}(v')$ measures the weight of $v'$ in $Cl$, and $s(v,v')$ measures the similarity between $v$ and $v'$. We compute $sim_{dom}(r, Cl)$ as the probability that they belong to the same group given several shared values as follows.

\vspace{-.15in} {\small
\begin{equation}
sim_{dom}(e, Cl)=1-\prod_{v \in e, v'\in ch}(1-p\cdot w_e(v) \cdot w_{Cl}(v') \cdot s(v, v')).
\end{equation}
} \vspace{-.1in}

\noindent
When there is no shared primary value, $sim_{dom}$ can be close to 0;
once there is one such value, $sim_{dom}$ can be significantly increased,
since we typically set a large $p$.

\begin{example}\label{ex:r-ch-sim}
Consider element $e=r_8$ and cluster $\mbox{Cl}_1=\{r_1-r_7\}$ in Example~\ref{ex:motivation}. 
Assume $p=.9$. Element $e$ and $\mbox{Cl}_1$ share the primary email domain, with
weight $1$ and ${5 \over 7}=.71$ respectively, but have different
phone numbers (assuming similarity of 0). We compute
$sim_{dom}(e,\mbox{Cl}_1)=1-(1-.9\cdot1\cdot.71\cdot1)\cdot (1-0) \cdot (1-0)\cdot (1-0)=.639$;
essentially, we do not penalize the difference in phone numbers.
Note however if {\em homedepot} appeared only once so was not a primary value, 
its weight would be $.14$ and accordingly $sim_{dom}(e,\mbox{Cl}_1)=.126$, 
indicating a much lower similarity. \rbox
\end{example}

Third, when we learn weights,
we learn one set of weights for distinct values (appearing in only one cluster) and one set for non-distinct values, such that distinct values, which can be considered as stronger evidence, typically contribute more to the final similarity.
In Example~\ref{ex:motivation}, sharing {\em ``Home Depot, The"} would serve
as stronger evidence than sharing {\em Taco Casa} for group similarity.

\subsection{Clustering algorithm}\label{sec:algo}
In most cases, clustering is intractable~\cite{Gonzalez82,clusternp}. We maximize the SV-index in a greedy fashion. Our algorithm starts with an initial clustering and then iteratively examines if we can improve the current clustering (increase SV-index) by merging clusters or moving elements between clusters. According to the definition of SV-index, in both initialization and adjusting, we always assign an element to the cluster with which it has the highest similarity.

\smallskip
\noindent {\bf Initialization}: Initially, we (1) assign each core to its own cluster and (2) assign a satellite $r$ to the cluster with the highest similarity if the similarity is above threshold $\theta_{ini}$ and create a new cluster for $r$ otherwise. We update the signature of each core along the way.
Note that initialization is sensitive in the order we consider
the records. Although designing an algorithm independent of the
ordering is possible, such an algorithm is more expensive and
our experiments show that the iterative adjusting can smooth out the difference.

\begin{example}\label{ex:initial}
Continue with the motivating example in Table~\ref{tab:motiv}.
First, consider records $r_{1}-r_{10}$, where $\mbox{Cr}_1=\{r_{1}-r_{7}\}$ is a core.
We first create a cluster $\mbox{Cl}_1$ for $\mbox{Cr}_1$. We then merge records $r_{8}-r_{10}$
to $\mbox{Cl}_1$ one by one, as they share similar names, and either
primary phone number or primary URL.

Now consider records $r_{11}-r_{20}$; recall that there are 2 cores and 5 satellites after core identification.
Figure~\ref{fig:clustering} shows the initialization result ${\cal C}_a$.
Initially we create two clusters $\mbox{Cl}_2, \mbox{Cl}_3$ for cores $\mbox{Cr}_2, \mbox{Cr}_3$.
Records $r_{11}, r_{19}-r_{20}$ do not share any primary value
on dominant-value attributes with $\mbox{Cl}_2$ or $\mbox{Cl}_3$,
so have a low similarity with them;
we create a new cluster for each of them.
Records $r_{12}$ and $r_{13}$ share the primary phone with $\mbox{Cr}_2$ so
have a high similarity; we link them to $\mbox{Cl}_2$. \rbox
\end{example}

\begin{figure}[t]
\centering
\includegraphics[scale=.35]{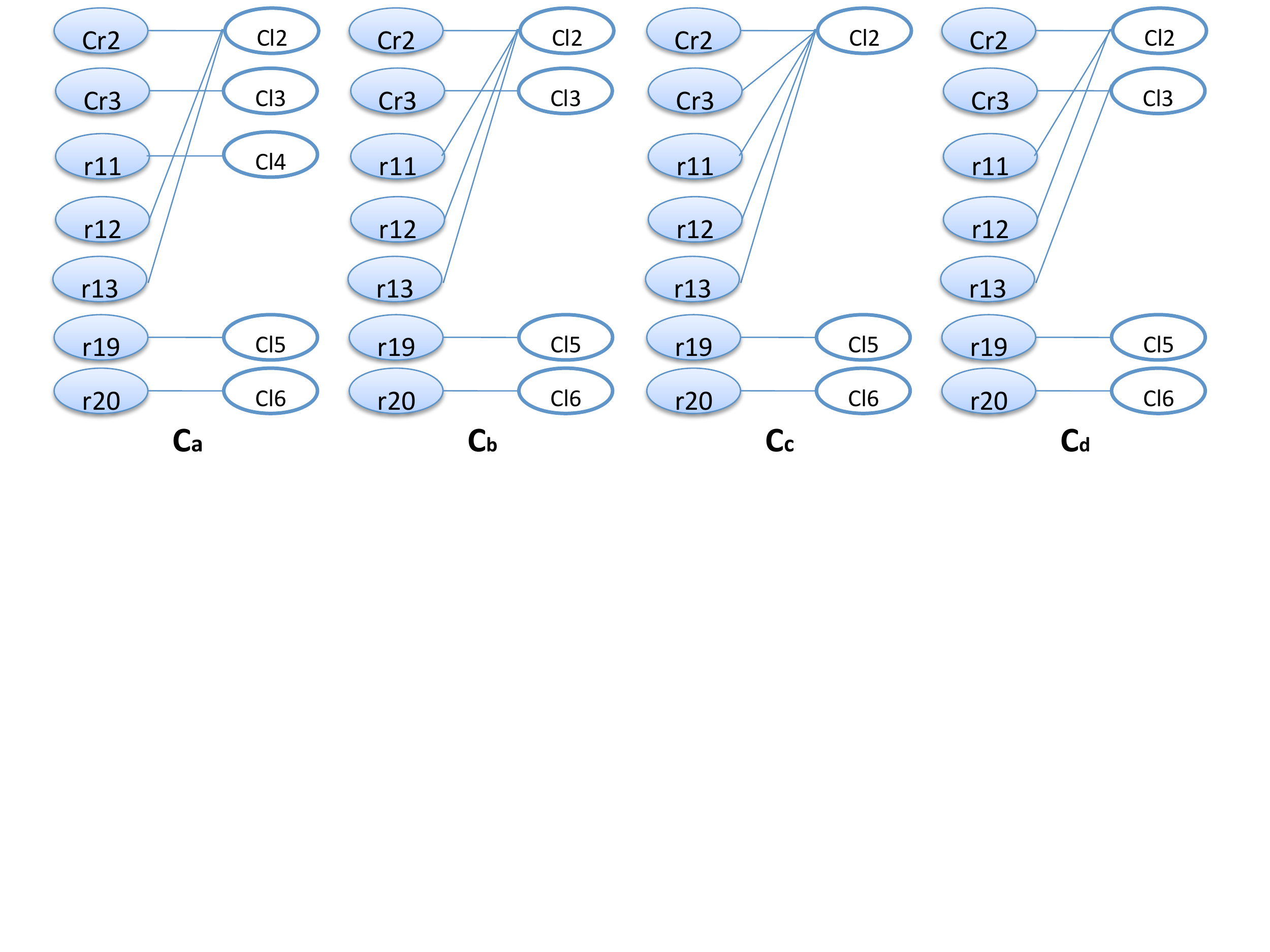}
\vspace{-1.65in} {\small\caption{Clustering of $r_{11}-r_{20}$ in Table~\ref{tab:motiv}.\label{fig:clustering}}}
\vspace{-.1in}
\end{figure}

\begin{table}
\vspace{-.1in}
 \scriptsize
  \centering
\caption{Element-cluster similarity and SV-index for clusterings in Figure~\ref{fig:clustering}.
Similarity between an element and its own cluster is in bold and
the second-to-highest similarity
is in italic. Low $S(e)$ scores are in italic.}
  \label{tbl:benefit}
  \begin{tabular}{|c|c|c|c|c|c||c|}
  \hline
  & $\mbox{Cl}_2$ & $\mbox{Cl}_3$ & $\mbox{Cl}_4$ & $\mbox{Cl}_5$& $\mbox{Cl}_6$& $S(e)$\\
  \hline
  $\mbox{Cr}_{2}$ & \textbf{.9} & \emph{.5} & \emph{.5}& \emph{.5} & \emph{.5} & .44\\
  $\mbox{Cr}_3$ & \emph{.6} & \textbf{1} & .5& .5 & .5 & .4\\
  $r_{11}$ & \emph{.7} & .5 & \textbf{1}& .5 & .5 & {\em .3}\\
  $r_{12}$ & \textbf{.99} & .5 & \emph{.95}& .5 & .5 & {\em .05}\\
  $r_{13}$ & \textbf{1} & .9 & \emph{.95}& .5 & .5 & {\em .05}\\
  $r_{19}$ & \emph{.5}& \emph{.5} & \emph{.5} & \textbf{1} & \emph{.5} & .5\\
  $r_{20}$ & \emph{.5}& \emph{.5} & \emph{.5} & \emph{.5} & \textbf{1} & .5\\
  \hline
\end{tabular}

(a) Cluster ${\cal C}_a$.

  \begin{tabular}{|c|c|c|c|c||c|}
  \hline
  & $\mbox{Cl}_2$ & $\mbox{Cl}_3$ & $\mbox{Cl}_5$ & $\mbox{Cl}_6$& $S(r)$\\
  \hline
  $\mbox{Cr}_2$ & \textbf{.87} & \emph{.5} & \emph{.5} & \emph{.5} & .43\\
  $\mbox{Cr}_3$ & \emph{.58}& \textbf{1} & .5 & .5 & .42\\
  $r_{11}$ & \textbf{.79}& \emph{.5} & \emph{.5} & \emph{.5} & .37\\
  $r_{12}$ & \textbf{.96}& \emph{.5} & \emph{.5} & \emph{.5} & .48\\
  $r_{13}$ & \textbf{.97}& \emph{.9} & .5 & .5 & {\em .07}\\
  $r_{19}$ & \emph{.5} & \emph{.5} & \textbf{1} & \emph{.5} & .5\\
  $r_{20}$ & \emph{.5} & \emph{.5} & \emph{.5} & \textbf{1} & .5\\
  \hline
\end{tabular}

(b) Cluster ${\cal C}_b$.
\vspace{-.15in}
\end{table}

\noindent {\bf Cluster adjusting:}
Although we always assign an element $e$ to the cluster with the highest
similarity so $S(e)>0$, the result clustering may still be improved by merging
some clusters or moving a {\em subset} of elements from one cluster
to another. Recall that when $S(e)$ is close to 0 and $a(e)$
is not too small, it indicates that
a pair of clusters might be similar and is a candidate for merging.
Thus, in cluster adjusting, we find such candidate pairs, iteratively
adjust them by merging them or moving a subset of elements between
them, and choose the new clustering if it increases the SV-index.

We first describe how we find candidate pairs. Consider element $e$
and assume it is closest to clusters $Cl$ and $Cl'$.
If $S(e)\leq \theta_s$, where $\theta_s$ is a threshold
for considering merging, 
 we call it a {\em border} element of $Cl$ and $Cl'$
and consider $(Cl, Cl')$ as a candidate pair.
We rank the candidates according to
(1) how many border elements they have
and (2) for each border element $e$, how close $S(e)$ is to 0.
Accordingly, we define the {\em benefit} of merging $Cl$ and $Cl'$
as $b(Cl,Cl')=\sum_{e\ is\ a\ border\ of\ Cl\ and\ Cl'}(1-S(e))$, and
rank the candidate pairs in decreasing order of the benefit.

We next describe how we re-cluster elements in a candidate
pair $(Cl, Cl')$. We adjust by merging the two
clusters, or moving the border elements between the clusters, or moving
out the border elements and merging them. Figure~\ref{fig:re-cluster}
shows the four re-clustering plans for a candidate pair.
Among them, we consider those that are valid
(\ie, a cluster cannot contain more than one satellite but no core)
and choose the one with the highest
SV-index. When we compute SV-index, we consider only elements
in $Cl, Cl'$ and those that are second-to-closest to
$Cl$ or $Cl'$ (their $a(e)$ or $b(e)$ can be changed)
such that we can reduce the computation cost.
After the adjusting, we need to re-compute $S(e)$ for these
elements and update the candidate-pair list accordingly.

\begin{figure}[t]
\centering
\includegraphics[scale=.35]{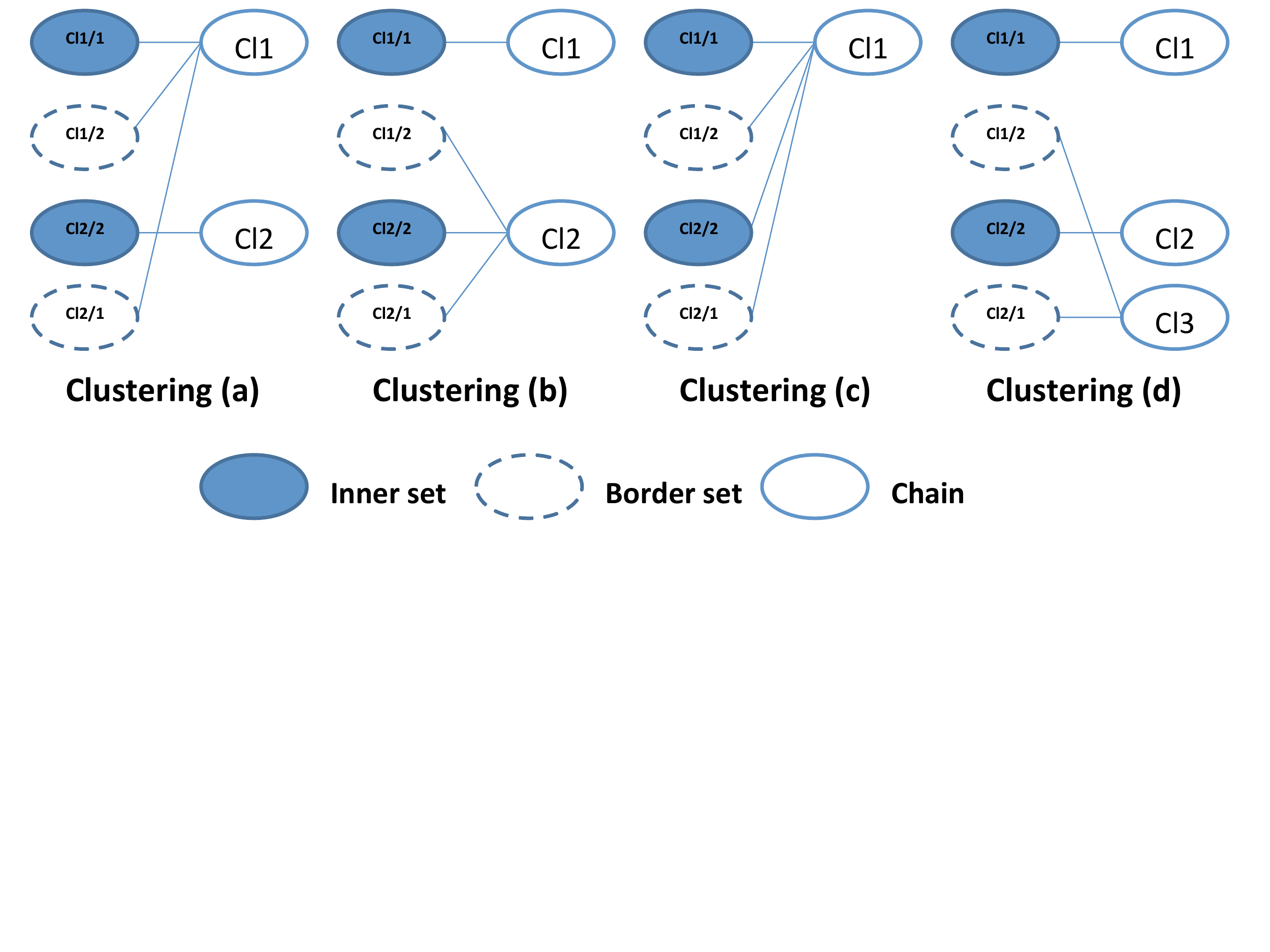}
\vspace{-1.45in} {\small\caption{Reclustering plans for $Cl_1$ and $Cl_2$.
\label{fig:re-cluster}}}
\vspace{-.1in}
\end{figure}

\begin{example}\label{ex:valid-clustering}
Consider adjusting cluster ${\cal C}_a$ in Figure~\ref{fig:clustering}.
Table~\ref{tbl:benefit}(a) shows similarity of each element-cluster
pair and SV-index of each element. Thus, the SV-index is .32.

Suppose $\theta_{s}=.3$. Then, $r_{11}-r_{13}$ are border elements
of $\mbox{Cl}_2$ and $\mbox{Cl}_4$, where $b(\mbox{Cl}_2,\mbox{Cl}_4)=.7+.95+.95=2.6$
(there is a single candidate so we do not need to compare
the benefit). For the candidate, we have two
re-clustering plans, $\{\{r_{11}-r_{13}, \mbox{Cr}_2\}\}$, $\{\{r_{11}-r_{13}\}, \{\mbox{Cr}_2\}\}$,
while the latter is invalid. For the former (${\cal C}_b$ in Figure~\ref{fig:clustering}),
we need to update $S(e)$ for every element and the new
SV-index is $.4$ (Table~\ref{tbl:benefit}(b)),
higher than the original one. \rbox
%
\end{example}

The full clustering algorithm {\sc Cluster}
 (details in Algorithm~\ref{alg:sate})
goes as follows.
\begin{enumerate}\tightlist
\item Initialize a clustering $\cal C$ and a list $Que$ of candidate
pairs ranked in decreasing order of merging benefit.
(Lines~\ref{ln:init1}-\ref{ln:init2}).
\item For each candidate pair $(Cl, Cl')$ in $Que$ do the following.

(a) Examine each valid adjusting plan
    and compute SV-index for it, and choose the one with
    the highest SV-index.  (Line~\ref{ln:compare}).

(b) Change the clustering if the new plan
    has a higher SV-index than the original clustering.
    Recompute $S(e)$ for each relevant element $e$ and move
    $e$ to a new cluster if appropriate.
    Update $Que$ accordingly.
    (Lines~\ref{ln:adjust1}-\ref{ln:adjust2}).
\item Repeat Step 2 until $Que=\emptyset$.
\end{enumerate}

{\small
\begin{algorithm}[t]
\caption{{\sc Cluster($\bf E, \theta_s$)}\label{alg:sate}
}
\begin{algorithmic}[1]
\REQUIRE $\bf E$: A set of cores and satellites for clustering.\\
\ \ $\theta_s$: Pre-defined threshold for considering merging.

\ENSURE $\cal C$: A clustering of elements in $\bf E$.

\STATE Initialize $\cal C$ according to $\bf E$;\label{ln:init1}

\STATE Compute $S({\cal C})$ and generate a list $Que$ of candidate pairs;\label{ln:init2}

\FOR{{\bf each} candidate pair $(Cl, Cl')\in Que$}

\STATE compute SV-index for its valid re-clustering plans and choose the clustering $\cal C$$_{max}$ with the highest SV-index;\label{ln:compare}

\IF{$S(\cal C$$)<S(\cal C$$_{max})$}

\STATE let ${\cal C}={\cal C}_{max}$, $change=true$;\label{ln:adjust1}

\WHILE{$change$}

\STATE $change=false$;

\FOR{{\bf each} relevant element $e$}

\STATE recompute $S(e)$;

\STATE When appropriate, move $e$ to a new cluster and set $change=true$;

\IF{$S(e)<\theta_s$ in the previous or current $\cal C$}

\STATE update the merging benefit of the related candidate pair and add it to $Que$ or remove it from $Que$ when appropriate;

\ENDIF

\ENDFOR

\ENDWHILE\label{ln:adjust2}

\ENDIF

\ENDFOR

\RETURN $\cal C$;
\end{algorithmic}
\end{algorithm}
}

\begin{proposition}\label{pro:cluster}
Let $l$ be the number of distinct candidate pairs ever in
$Que$ and $|{\bf E}|$ be the number of input elements.
Algorithm {\sc Cluster} takes time $O(l\cdot|{\bf E}|^2)$.\rbox
\end{proposition}
\begin{proof}\label{proof:cluster}
It takes time $O(|{\bf E}|^2)$ to initialize clustering $\cal C$ and list $Que$. It takes $|{\bf E}|^2$ to check each distinct candidate pair in $Que$, where it takes $O(|{\bf E}|)$ to examine all valid clustering plans and select the one with highest SV-index (Step 2(a)), and it takes $O(|{\bf E}|^2)$ to recompute SV-index for all relevant elements and update $Que$ (Step 2(b)). In total there are $l$ distinct candidate pairs ever in $Que$, thus {\sc Cluster} takes time $O(l\cdot|{\bf E}|^2)$.
\end{proof}

Note that we first block records according to name similarity
and take each block as an input, so typically $|{\bf E}|$ is quite small.
Also, in practice we need to consider only a few
candidate pairs for adjusting in each input, so $l$ is also small.

\begin{example}\label{ex:cluster}
Continue with Example~\ref{ex:valid-clustering} and
consider adjusting ${\cal C}_b$.
Now there is one candidate pair $(\mbox{Cl}_2, \mbox{Cl}_3)$,
with border $r_{13}$.
We consider clusterings ${\cal C}_c$ and ${\cal C}_d$.
Since $S(\cal C$$_c)=.37<.40$ and $S(\cal C$$_d)=.32<.40$,
we keep $\cal C$$_b$ and return it as the result.
We do not merge records $\mbox{Cl}_2=\{r_{11}-r_{15}\}$
with $\mbox{Cl}_3=\{r_{16}-r_{18}\}$,
because they share neither phone nor the
primary URL. {\sc Cluster} returns the correct chains.\rbox
\end{example}

\section{Experimental Evaluation}\label{sec:experiment}
This section describes experimental results on two real-world
data sets, showing high scalability of our techniques, and advantages
of our algorithm over rule-based or traditional machine-learning
methods on accuracy. 

\begin{table}[t]
\vspace{-.1in}
 \scriptsize
  \centering
\caption{Statistics of the experimental data sets.}
  \label{tbl:data}
  \begin{tabular}{|c|c|c|c|c|}
  \hline
  & & $\#$Groups & & \#Singletons\\
  & \raisebox{1.5ex}[0pt]{$\#$Records} & (size $>1$) & \raisebox{1.5ex}[0pt]{Group size} & (size $=1$)  \\
  \hline
  {\em Random} & 2062 &30 &[2, 308] & 503 \\
  \hline
  {\em AI} & 2446 & 1 & 2446 & 0 \\
  \hline
  {\em UB} &322 & 9 & [2, 275] & 5 \\
  \hline
  {\em FBIns} &1149 & 14 & [33, 269] & 0 \\
  \hline
  {\em SIGMOD} & 590& 71 & [2, 41] & 162 \\
  \hline
\end{tabular}
\vspace{-.1in}
\end{table}
\subsection{Experiment settings}\label{sec:setting}
\noindent {\bf Data and gold standard:}
We experimented on two real-world data sets. {\em Biz} contains
18M US business listings 
and each listing has attributes {\sf name, phone, URL, location}
and {\sf category}; we decide which listings belong to the
same business chain. {\em SIGMOD} contains records about
590 attendees of SIGMOD'98 and each record has attributes
{\sf name, affiliation, address, phone, fax} and {\sf email};
we decide which attendees belong to the same institute.

We experimented on the whole {\em Biz} data set to study scalability
of our techniques. We evaluated accuracy of our techniques on
five subsets of data. The first four are from {\em Biz}.
(1) {\em Random} contains 2062 listings from {\em Biz},
where 1559 belong to 30 randomly selected business chains,
and 503 do not belong to any chain; among the 503 listings, 86 are highly
similar in {\sf name} to listings in the business chains and the
rest are randomly selected. (2) {\em AI} contains
2446 listings for the same business chain {\em Allstate Insurance}.
These listings have the same name, but
1499 provide URL {\em ``allstate.com''},
854 provide another URL {\em ``allstateagencies.com''},
while 130 provide both, and 227 listings do not
provide any value for {\sf phone} or {\sf URL}.
(3) {\em UB} contains 322 listings with exactly the same
name {\em Union Bank} and highly similar category values;
317 of them belong to
9 different chains while 5 do not belong to any chain.
(4) {\em FBIns data set} contains 1149 listings with similar names
and highly similar category
values; they belong to 14 different chains. Among the listings,
708 provide the same wrong name {\em Texas Farm Bureau Insurance}
and meanwhile provide a wrong {\em URL}
{\em farmbureauinsurance-mi.com}.
Among these four subsets, the latter three are hard cases;
for each data set, we manually verified all the chains by checking
store locations provided by the business-chain websites
and used it as the gold standard.
The last ``subset'' is actually the whole {\em SIGMOD} data set.
It has very few wrong values, but the same affiliation
can be represented in various ways and some affiliation
names can be very similar (e.g., {\em UCSC} vs. {\em UCSD}).
We manually identified 71 institutes that have multiple
attendees and there are 162 attendees who do not
belong to these institutes.
Table~\ref{tbl:data} shows statistics of the five subsets.

\smallskip
\noindent {\bf Measure:} We considered each group as
a cluster and compared pairwise linking decisions with the
gold standard. We measured the quality of the results
by \emph{precision} ($P$), \emph{recall} ($R$), and \emph{F-measure} ($F$).
If we denote the set of true-positive pairs by $TP$,
the set of false-positive pairs by $FP$,
and the set of false-negative pairs by $FN$,
then, $P=\frac{|TP|}{|TP|+|FP|}$, $R=\frac{|TP|}{|TP|+|FN|}$, $F=\frac{2PR}{P+R}$.
In addition, we reported execution time.

\smallskip
\noindent {\bf Implementation:}
We implemented the technique we proposed in this paper,
and call it {\sc Group}. In core generation, for {\em Biz}
we considered two records are similar
if (1) their name similarity is above .95; and (2) they share
at least one phone or URL domain name. For {\em SIGMOD}
we require (1) affiliation similarity is above .95; and
(2) they share at least one of phone prefix (3-digit),
fax prefix (3-digit), email server, or the addresses have
a similarity above .9. We required 2-robustness
for cores. In clustering, (1) for blocking,
we put records whose name similarity
is above .8 in the same block; (2) for similarity computation,
we computed string similarity by Jaro-Winkler distance~\cite{crf03},
we set $\alpha = .01, \beta = .02,
\theta_{th}=.6, p = .8$, and we learned other weights from 1000 records
randomly selected from {\em Random} data for {\em Biz},
and 300 records randomly selected from {\em SIGMOD}.
We discuss later the effect of these choices.

For comparison, we also implemented the following baselines:
\begin{itemize}\tightlist
\item  {\sc SameName} groups {\em Biz} records with highly
similar names and groups {\em SIGMOD} records with highly
similar affiliations (similarity above .95);
\item  {\sc ConnectedGraph} generates the similarity graph as\\
{\sc Group} but considers each connected subgraph as a group;
\item One-stage machine-learning linkage methods include 
{\sc Partition}, {\sc Center} and {\sc Merge}~\cite{Hassanzadeh09frameworkfor};
each method computes record similarity by
Eq.(\ref{eqn:sim}) with learned weights.
\item Two-stage method
{\sc Yoshida}~\cite{Yoshida:2010:PND:1835449.1835454} generates cores by
agglomerative clustering with threshold .9 in the first stage,
uses TF/IDF weights for features and applies linear algebra to
assign each record to a group in the second stage.
\end{itemize}

We implemented the algorithms in Java. We used a Linux
machine with Intel Xeon X5550 processor (2.66GHz,
cache 8MB, 6.4GT/s QPI).
We used MySQL to store the data sets and stored the index
as a database table. Note that after blocking, we can fit each block of 
nodes or elements into main memory, which is typically the case
with a good blocking strategy.

\begin{figure}[t]
\vspace{-.1in}
\centering
\includegraphics[scale=.35]{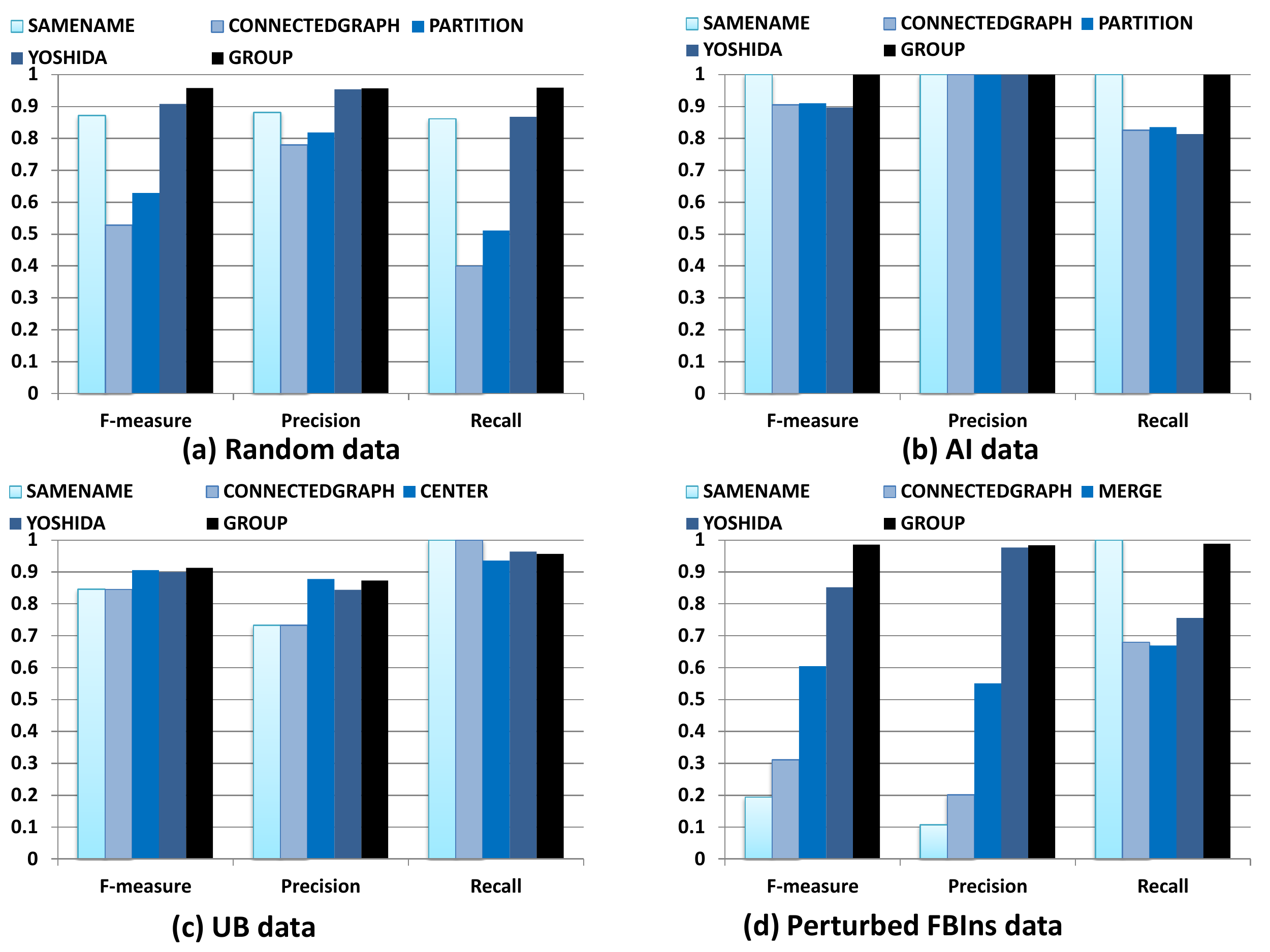}
\vspace{-0.24in}
{\small\caption{Overall results on {\em Biz} data set.\label{fig:overall}}}
\vspace{-.1in}
\end{figure}

\begin{figure}[t]
\centering
\includegraphics[scale=.35]{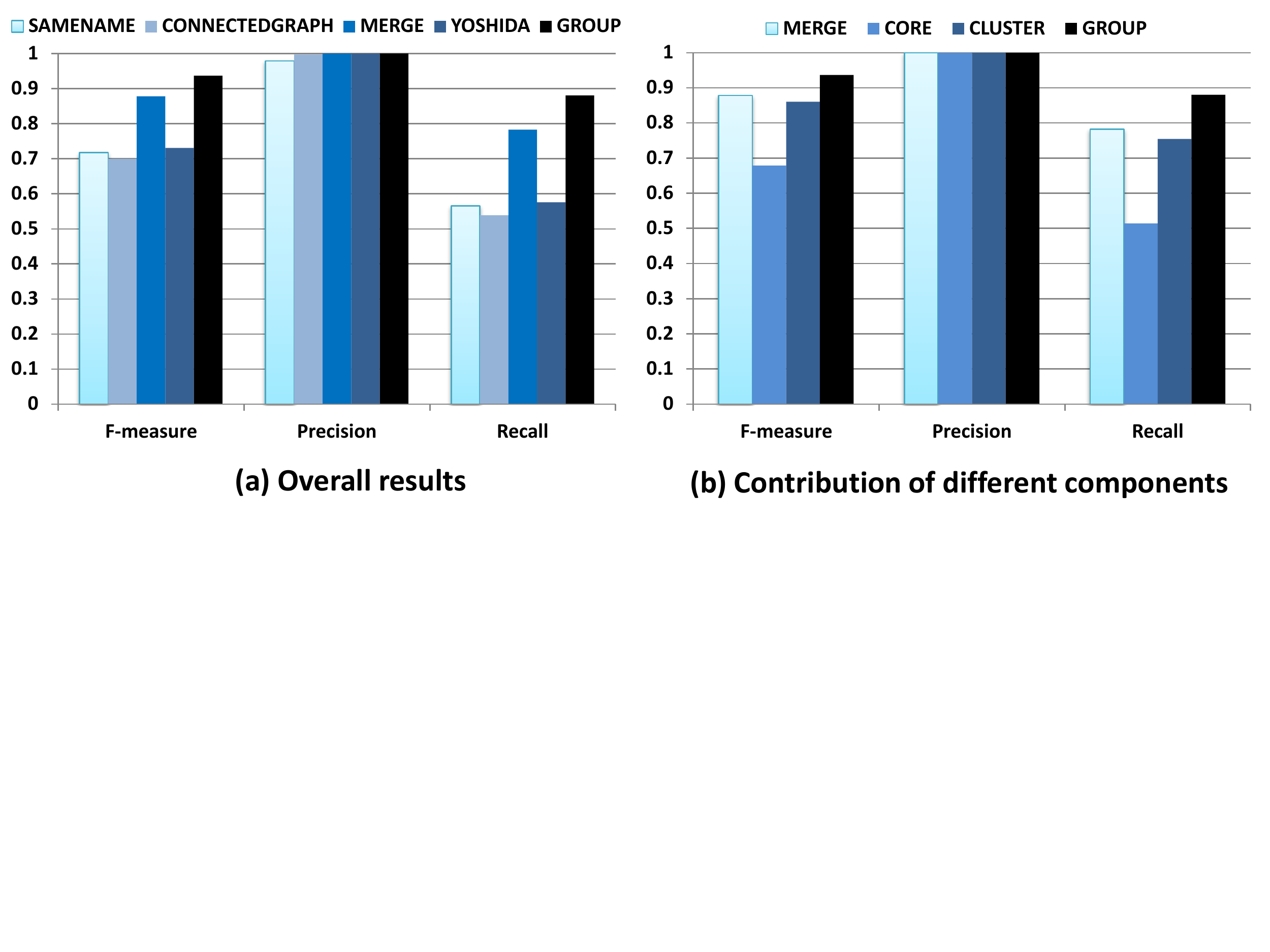}
\vspace{-1.5in}
{\small\caption{Results on {\em SIGMOD} data.\label{fig:sigmod_overall}}}
\vspace{-.2in}
\end{figure}

\subsection{Evaluating effectiveness}
\label{sec:effe}
We first evaluate effectiveness of our algorithms.
Figure~\ref{fig:overall} and Figure~\ref{fig:sigmod_overall}(a) compare {\sc Group} with the baseline methods, where for the three one-stage linkage methods
we plot only the best results. On {\em FBIns}, all methods put
all records in the same chain because a large number (708) of listings
have both a wrong name and a wrong URL. We manually perturbed the data
as follows: (1) among the 708 listings with wrong URLs,
408 provide a single (wrong) URL and we fixed it;
(2) for all records we set {\sf name} to {\em ``Farm Bureau
Insurance''}, so removed hints from business names.
Even after perturbing, this data set remains the hardest
and we use it hereafter instead of the original one for other experiments.

We have the following observations. (1) {\sc Group}
obtains the highest F-measure (above .9) on each data set.
It has the highest precision most of the time
as it applies core identification
and leverages the strong evidence collected from resulting cores.
It also has a very high recall (mostly above .95) on each subset
because the clustering phase is tolerant to diversity of
values within chains.
(2) The F-measure of {\sc SameName} is 7-80\% lower than {\sc Group}.
It can have false positives
when listings of highly similar names belong to different chains
and can also have false negatives when some listings in a chain
have fairly different names from other listings. It only
performs well in {\em AI}, where it happens that all listings
have the same name and belong to the same chain.
(3) The F-measure of {\sc ConnectedGraph} is 2-39.4\% lower than {\sc SameName}.
It requires in addition sharing at least
one value for dominant-value attributes. As a result, it has a lower recall
than {\sc SameName}; it has fewer false positives than
{\sc SameName}, but because it has fewer true positives,
its precision can appear to be lower too. (4) The highest F-measure of 
one-stage linkage methods is 1-94.7\% higher than {\sc ConnectedGraph}.
As they require high record similarity, it has similar number of
false positives to {\sc ConnectedGraph} but often has much more true positives;
thus, it often has a higher recall and also a higher precision.
However, the highest F-measure is still 1-38.7\% lower than {\sc Group}.
(5) {\sc Yoshida} has comparable precision to {\sc Group}
since its first stage is conservative too, which makes it often improve over the best of one-stage linkage methods on \emph{Biz} dataset where reducing false positives is a big challenge; on the other hand,
its first stage is often too conservative (requiring high
record similarity) so the recall is 10-34.6\% lower than {\sc Group}, which also makes it perform worse than one-stage linkage methods on \emph{Sigmod} dataset where reducing false negatives is challenging. 

\eat{
On the {\em FBIns} data set,
data set and Figure~\ref{fig:overall}(d) shows the results.
We observe that almost baseline methods have very low F-measure
(below .6) while {\sc Group} still obtains a F-measure
as high as .98. {\sc 2-Stage RecordLinkage} does not perform as well as {\sc Group} as in the other {\em business} data sets.

The {\em SIGMOD} data set has fewer wrong values in each attribute, therefore we observe that all algorithms have similar high precision (above .97). The most common problem we found in the data is multiple representations of the same value. We observe that all baseline algorithms are not as tolerant with such heterogeneity as {\sc Group} thus have lower recalls.
}

\begin{figure}[t]
\vspace{-.1in}
\centering
\includegraphics[scale=.35]{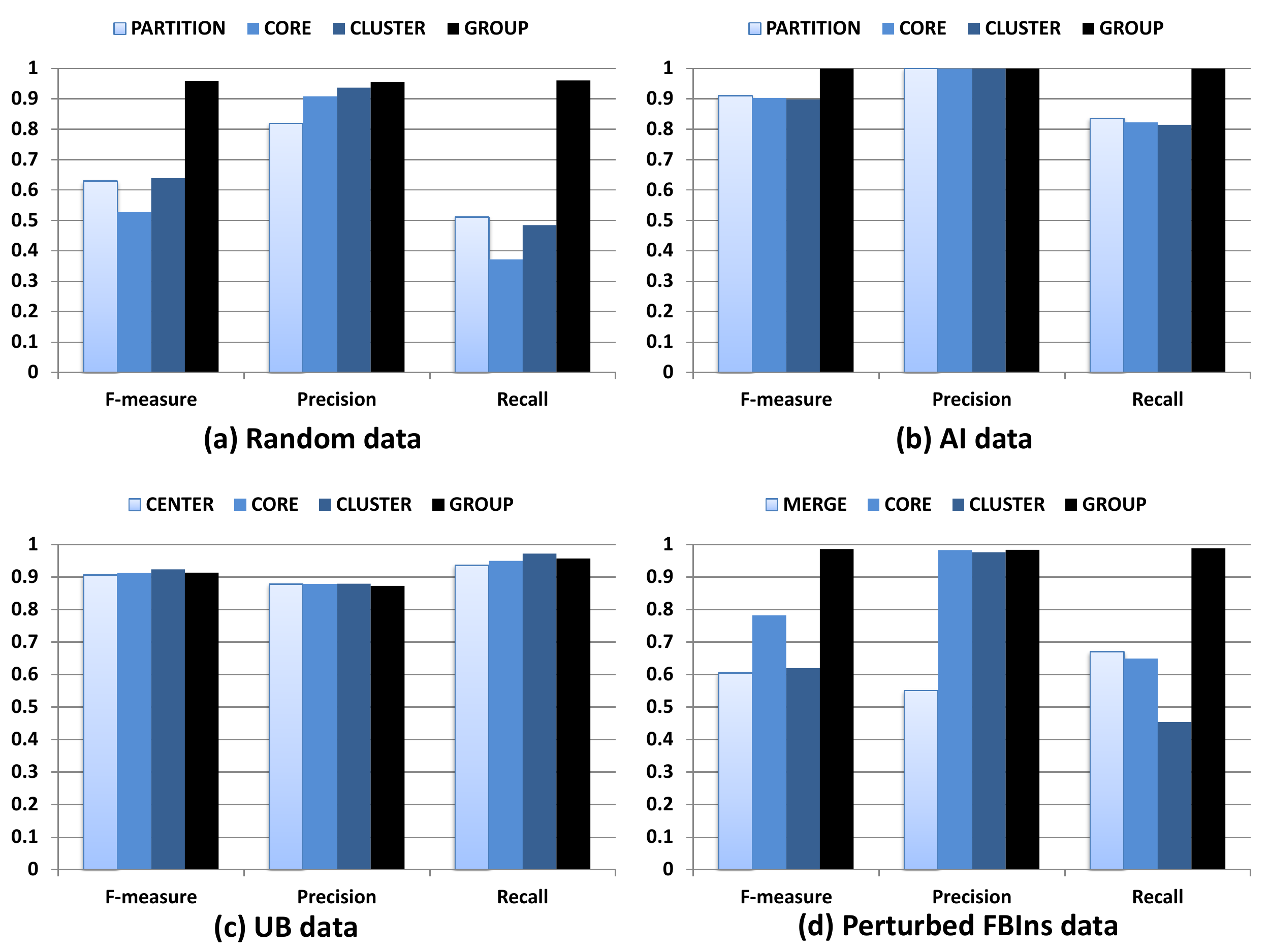}
\vspace{-0.2in}
{\small\caption{Contribution of different components on {\em Biz}.\label{fig:component}}}
\vspace{-.12in}
\end{figure}

\smallskip
\noindent
{\bf Contribution of different components:}
We compared {\sc Group} with (1) {\sc Core}, which
applies Algorithm {\sc CoreIdentification} but does not
apply clustering, and (2) {\sc Cluster}, which considers
each individual record as a core and applies Algorithm {\sc Cluster}
(in the spirit of~\cite{Larsen:1999:FET:312129.312186, Clustering07ricochet:a}).
Figure~\ref{fig:component} and Figure~\ref{fig:sigmod_overall}(b) show the results. First, we observe that {\sc Core} improves over one-stage linkage
methods on precision by .1-78.6\% but has a lower recall (1.5-34.3\% lower) most of the time,
because it sets a high requirement for merging records into groups.
Note however that its goal is indeed to obtain a high precision
such that the strong evidence collected from the cores are
trustworthy for the clustering phase.
Second, {\sc Cluster} often has higher precision (by 1.6-77.3\%) but lower
recall (by 2.5-32.2\%) than the best one-stage linkage methods; their F-measures
are comparable on each data set.
On some data sets ({\em Random, FBIns}) it can obtain an even
higher precision than {\sc Core}, because
{\sc Core} can make mistakes when too many records have
erroneous values, but {\sc Cluster} may avoid some of these
mistakes by considering also
similarity on {\sf state} and {\sf category}.
However, applying clustering on the results
of {\sc Cluster} would not change the results,
but applying clustering on the
results of {\sc Core} can obtain a much higher F-measure,
especially a higher recall (98\% higher than {\sc Cluster}
on {\em Random}). This is because the result of {\sc Cluster}
lacks the strong evidence collected from high-quality
cores so the final results would be less tolerant to diversity of values,
showing the importance of core identification.
Finally, we observe that {\sc Group} obtains the best
results in most of the data sets.

We next evaluate various choices in the two stages.
Unless specified otherwise, we observed similar patterns
on each data set from {\em Biz} and {\em Sigmod},
and report the results on {\em Random}
or perturbed {\em FBIns} data, whichever has more distinguishable results.

\subsubsection{Core identification}
\label{sec:exp-core}
\begin{figure}[t]
\centering
\includegraphics[scale=.35]{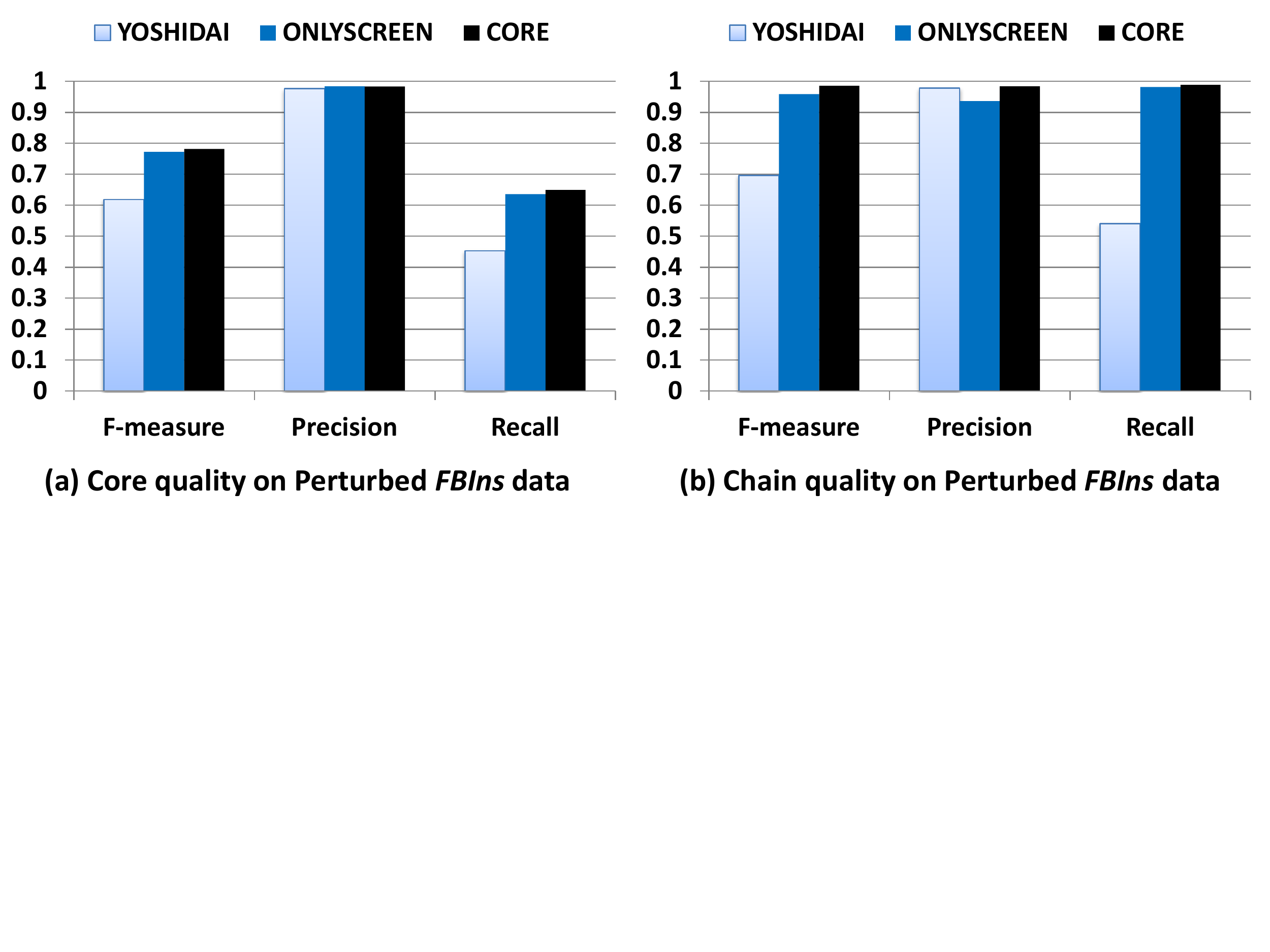}
\vspace{-1.5in}
{\small\caption{Core identification on perturbed {\em FBIns} data.\label{fig:fbi_red}}}
\vspace{-.1in}
\end{figure}
\noindent{\bf Core identification:} We first compared three core-generation
strategies: {\sc Core} iteratively invokes {\sc Screen}
and {\sc Split}, {\sc OnlyScreen} only iteratively invokes {\sc Screen},
and {\sc YoshidaI} generates cores by agglomerative clustering~\cite{Yoshida:2010:PND:1835449.1835454}.
Recall that by default we apply {\sc Core}.
Figure~\ref{fig:fbi_red} compares them on the perturbed {\em FBIns} data.
First, we observe similar results of {\sc OnlyScreen} and {\sc Core} on all data sets since most inputs to {\sc Split} pass the $k$-robustness test. Thus, although {\sc Screen} in itself cannot guarantee soundness of the resulting cores (k-robustness), it already does well in practice.
Second, {\sc YoshidaI} has lower recall in both core and clustering results,
since it has stricter criteria in core generation.

\begin{figure}[t]
\centering
\includegraphics[scale=.35]{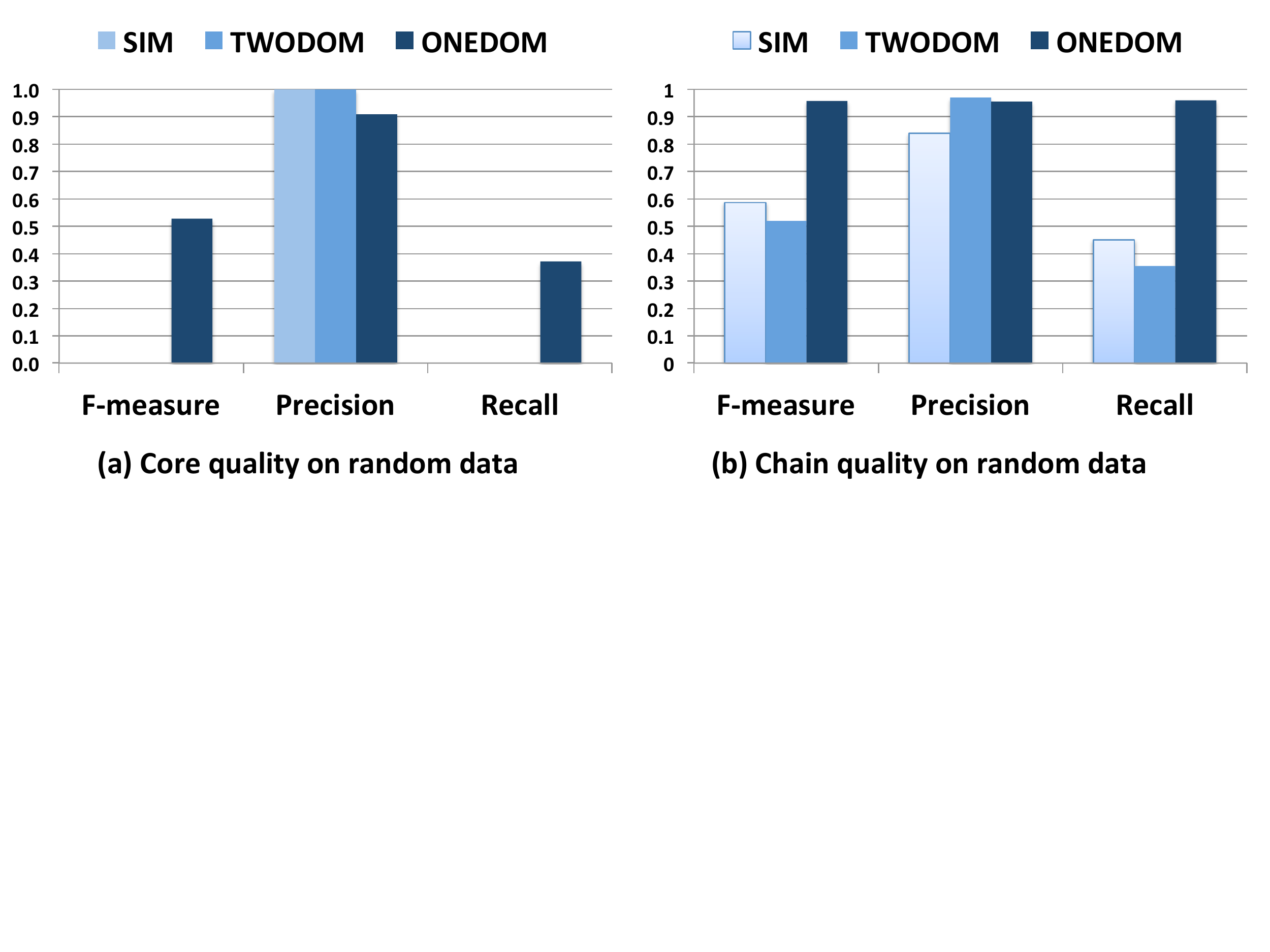}
\vspace{-1.5in}
{\small\caption{Effect of graph generation on {\em Random} data.\label{fig:cores}}}
\vspace{-.11in}
\end{figure}
\smallskip
\noindent{\bf Graph generation:} We compared three edge-adding
strategies for similarity graphs: {\sc Sim} takes weighted similarity
on each attribute except {\sf location} and requires
a similarity of over .8;
{\sc TwoDom} requires sharing {\sf name} and at least two
values on dominant-value attributes;
{\sc OneDom} requires
sharing {\sf name} and one value on dominant-value attributes.
Recall that by default we applied {\sc OneDom}.
Figure~\ref{fig:cores} compares these three strategies.
We observe that
(1) {\sc Sim} requires similar records so has a high precision,
with a big sacrifice on recall for the cores (0.00025);
as a result, the F-measure of the chains is very low (.59);
(2) {\sc TwoDom} has the highest requirements
and so even lower recall than {\sc Sim} for the cores (.00002),
and in turn it has the lowest F-measure for the chains (.52).
This shows that only requiring high precision for cores with
big sacrifice on recall can also lead to low F-measure for the
chains.

We also varied the similarity requirement for names
and observed very similar results (varying by .04\%)
when we varied the threshold from .8 to .95.

\begin{figure}[t]
\centering
\includegraphics[scale=.35]{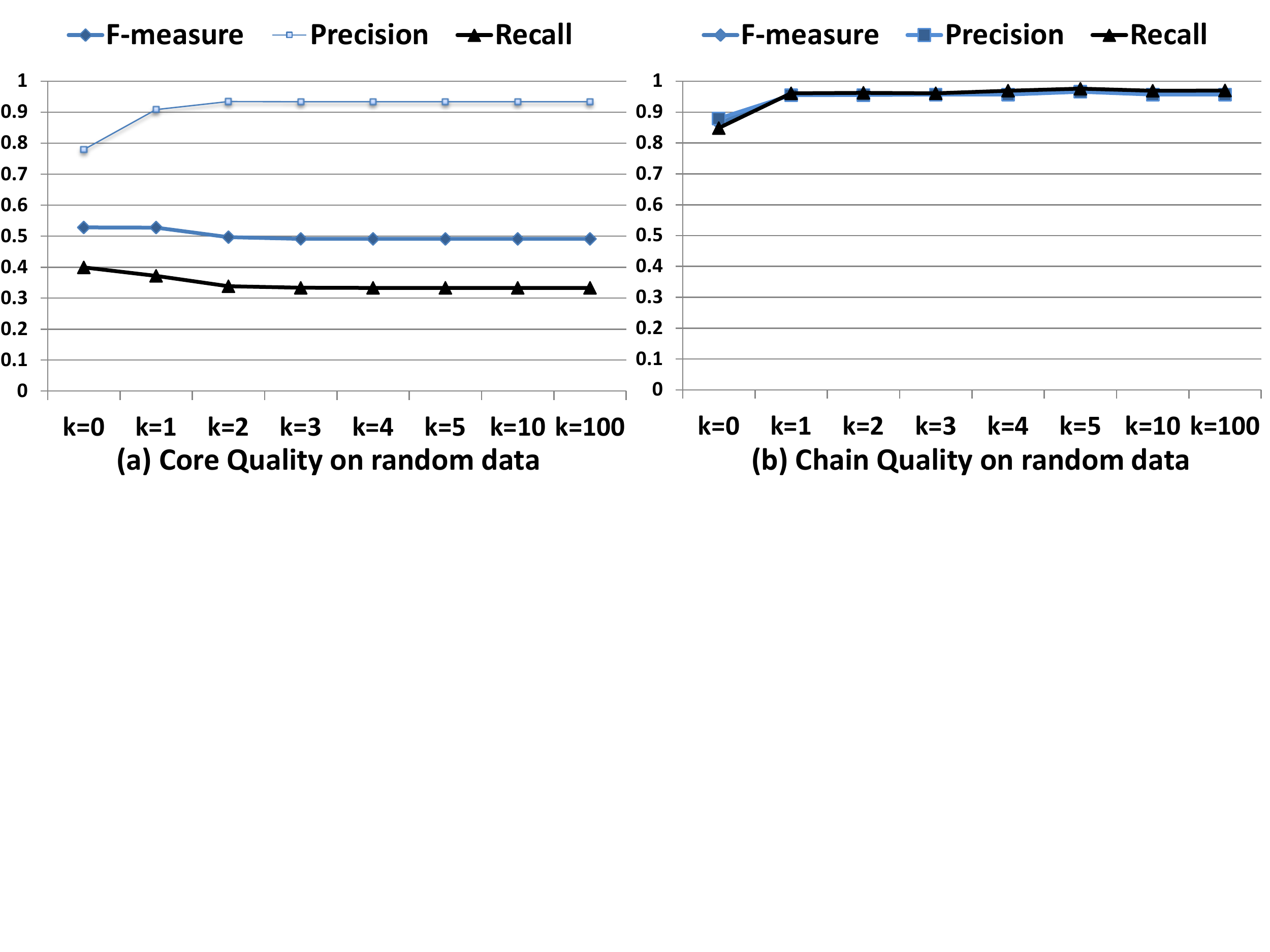}
\vspace{-1.5in}
{\small\caption{Effect of robustness requirement on {\em Random} data.\label{fig:fbi_robust}}}
\vspace{-.1in}
\end{figure}
\smallskip
\noindent
{\bf Robustness requirement:} We next studied how the robustness
requirement can affect the results (Figure~\ref{fig:fbi_robust}).
We have three observations.
(1) When $k=0$, we essentially take every connected subgraph
as a core, so the generated cores can have a much
lower precision; those false positives cause both a low precision
and a low recall for the resulting chains because we do not collect
high-quality strong evidence.
(2) When we vary $k$ from 1 to 4, the number of false positives
decreases while that of false negatives increases for the cores,
and the F-measure of the chains increases but only very slightly.
(3) When we continue increasing k, the results of cores and clusters remain stable. This is because setting k=4 already splits the graph into subgraphs, each containing a single v-clique, so further increasing k would not change the cores. This shows that considering $k$-robustness is
important, but $k$ does not need to be too high.

%

\subsubsection{Clustering}

\begin{figure}[t]
\begin{center}
\includegraphics[width=8cm]{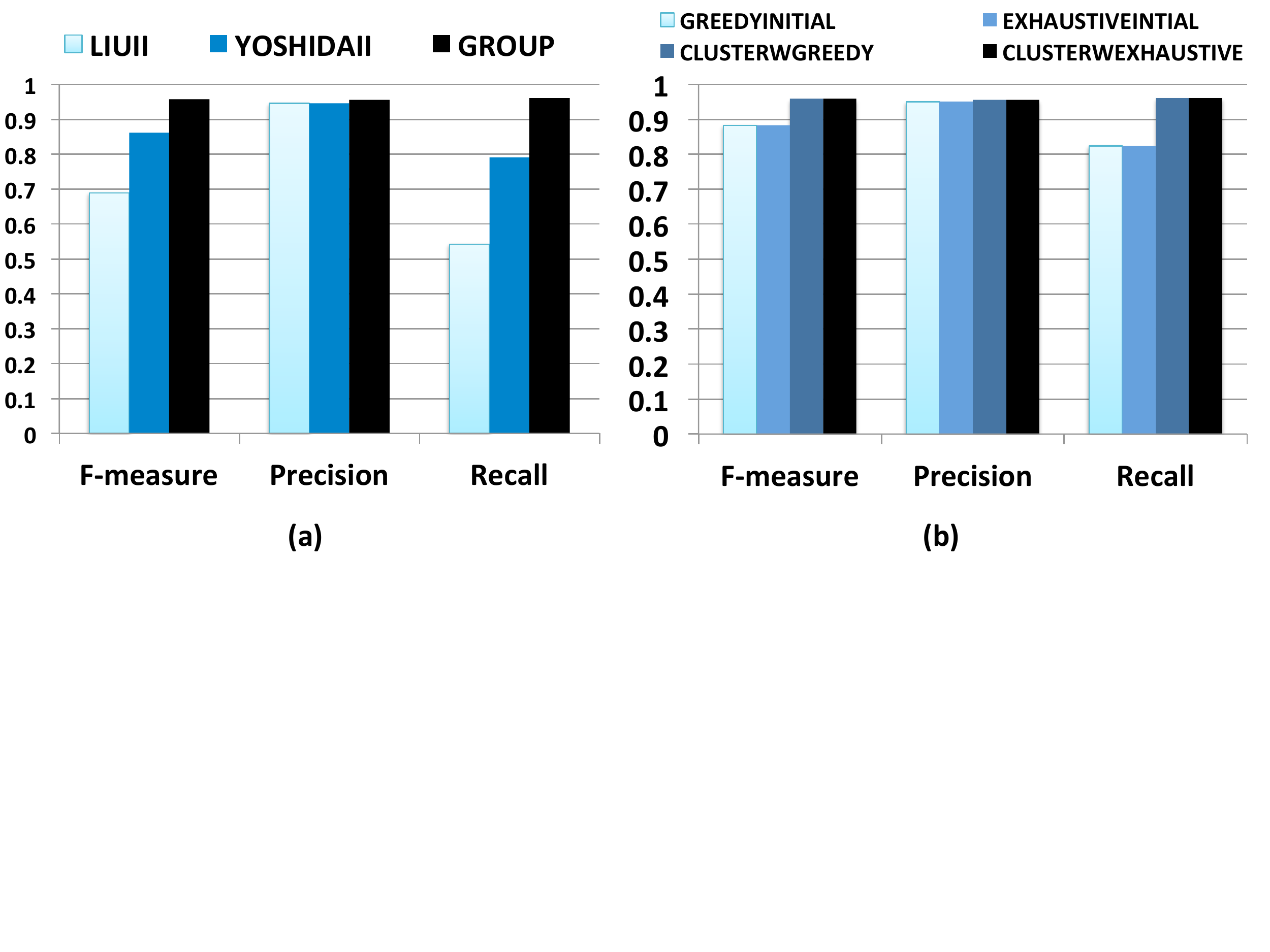}

\vspace{-1.1in}
{\small\caption{Clustering strategies on {\em Random} data.\label{fig:chain_comp}}}
\end{center}
\vspace{-.25in}
\end{figure}
\noindent{\bf Clustering strategy}: We first compared our clustering
algorithm with two algorithms proposed for the second stage of
two-stage clustering: {\sc LiuII}~\cite{Liu:2002:DCC:564376.564411}
iteratively applies majority voting to assign each record to
a cluster and collects a set of representative features for each cluster using a threshold
(we set it to 5, which leads to the best results);
{\sc YoshidaII}~\cite{Yoshida:2010:PND:1835449.1835454} is
the second stage of {\sc Yoshida}.
Figure~\ref{fig:chain_comp}(a) compares their results.
We observe that our clustering method improves the recall by
39\% over {\sc LiuII} and by 11\% over {\sc YoshidaII}.
{\sc LiuII} may filter strong evidence by the threshold;
{\sc YoshidaII} cannot handle records whose dominant-value
attributes have null values well.

\eat{(1) starting from the result of {\sc Core}, our clustering algorithm improves F-measure over baseline methods by up to 39\%, (2) applying cluster adjusting can improve
the F-measure a lot (by 8.6\%), and (2) exhaustive initialization
does not significantly improve over greedy initialization, if at all.
This shows effectiveness of the current algorithm {\sc Cluster}.
}

We also compared four clustering
algorithms: {\sc GreedyInitial} performs only initialization
as we described in Section~\ref{sec:sate};
{\sc ExhaustiveInitial} also performs only initialization, but
by iteratively conducting matching and merging
until no record can be merged
to any core; {\sc ClusterWGreedy} applies cluster adjusting
on the results of {\sc GreedyInitial}, and
{\sc ClusterWExhaustive} applies cluster adjusting
on the results of {\sc ExhaustiveInitial}.
Recall that by default we apply {\sc ClusterWGreedy}.
Figure~\ref{fig:chain_comp}(b) compares their results.
We observe that (1) applying cluster adjusting can improve
the F-measure a lot (by 8.6\%), and (2) exhaustive initialization
does not significantly improve over greedy initialization, if at all.
This shows effectiveness of the current algorithm {\sc Cluster}.


\begin{figure}[t]
\begin{center}
\begin{minipage}{.49\linewidth}
\begin{center}
\includegraphics[width=8cm]{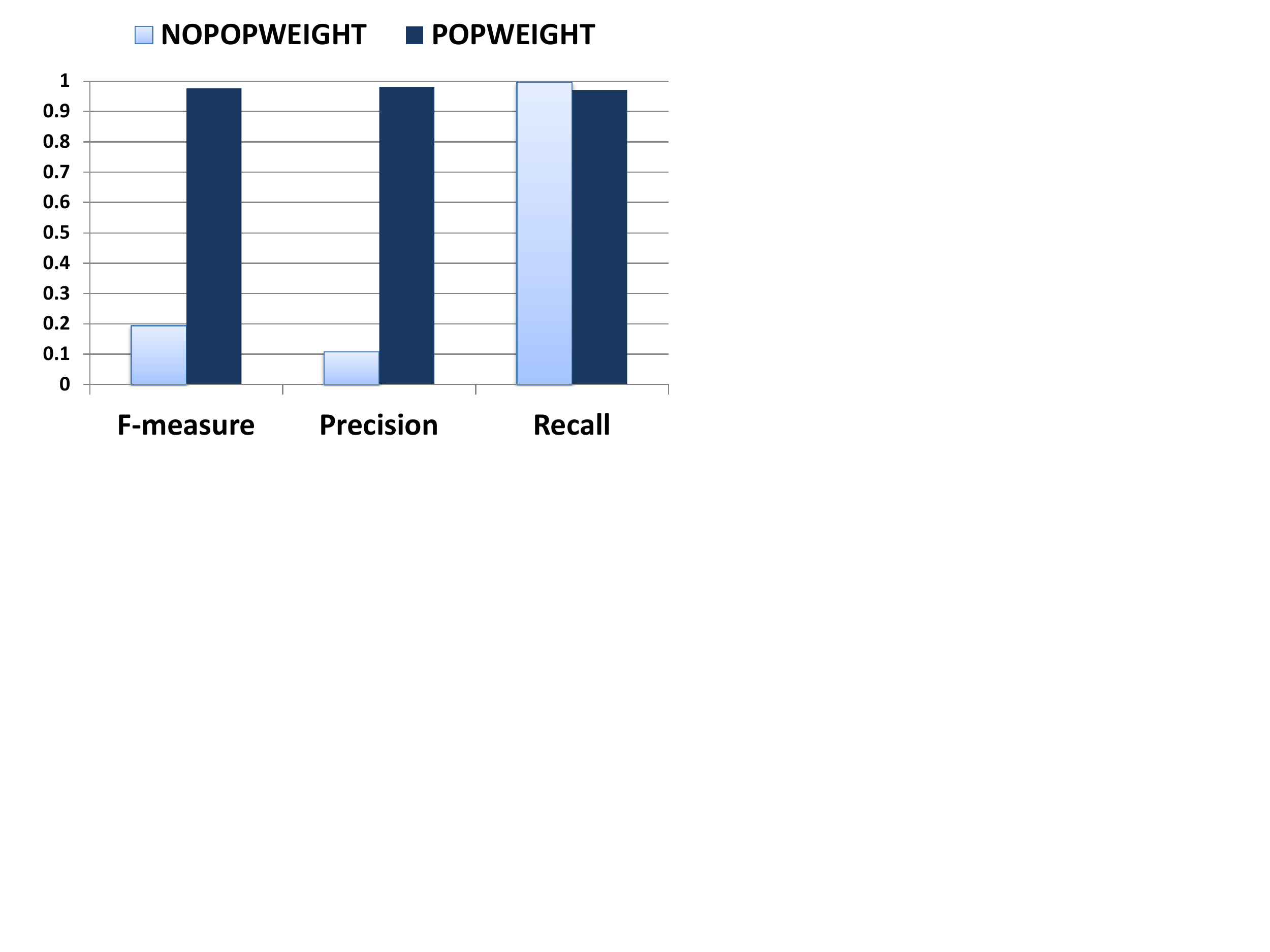}
\vspace{-1.5in}
{\small\caption{Value weights on perturbed {\em FBIns} data.\label{fig:popularity}}}
\end{center}
\end{minipage}
\hfill
\begin{minipage}{.49\linewidth}
\begin{center}
\includegraphics[width=8cm]{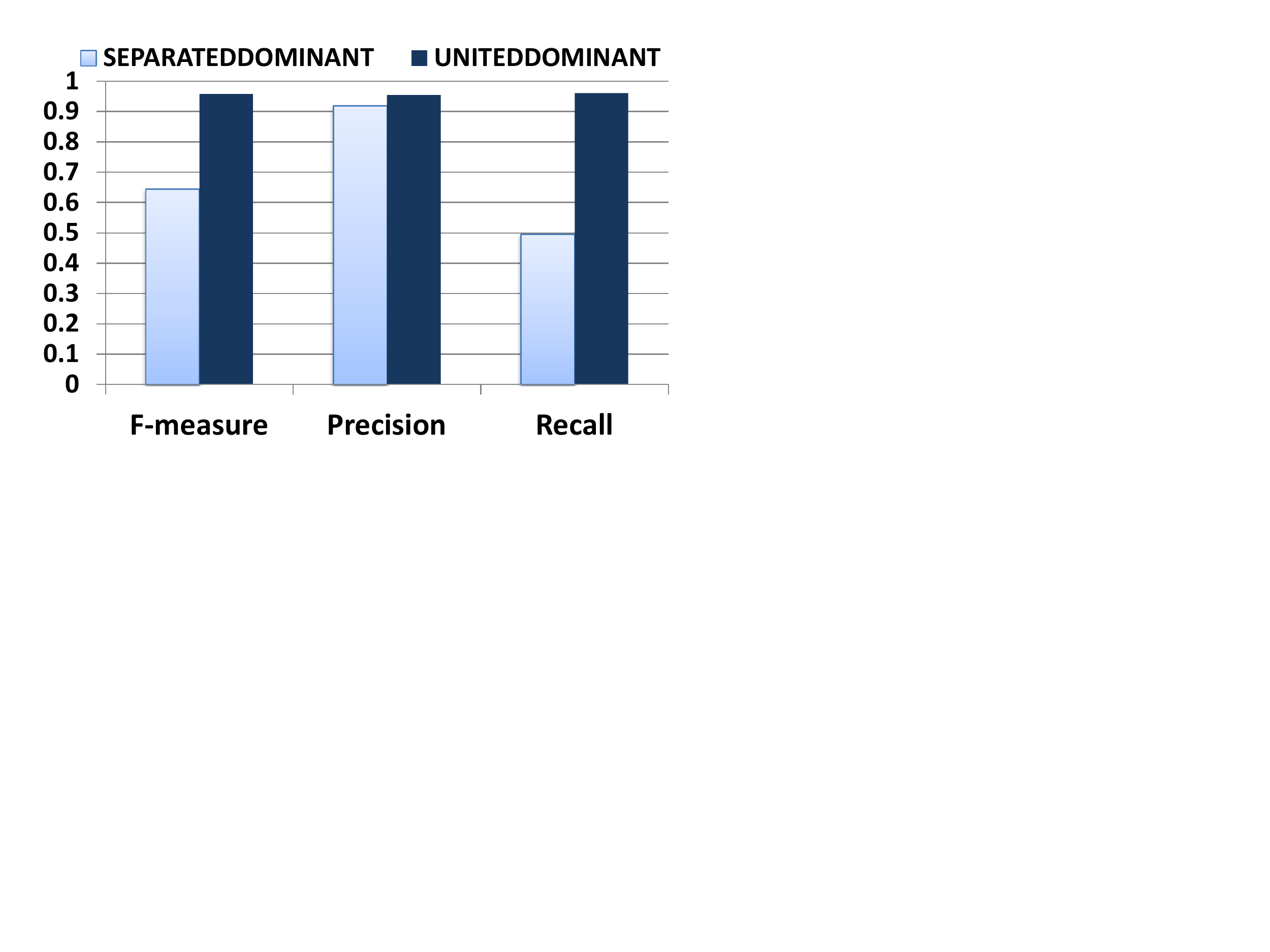}
\vspace{-1.5in}
{\small\caption{Dominant-value attributes on {\em Random}.\label{fig:dmnt}}}
\end{center}
\end{minipage}
\end{center}
\vspace{-.2in}
\end{figure}
\smallskip
\noindent{\bf Value weight:} We then compared the results with and without
setting popularity weights for values.
Figure~\ref{fig:popularity} compares the results with and without
setting popularity weights on perturbed {\em FBIns} data.
We observe that setting the popularity weight helps
distinguish primary values from unpopular values,
thus can improve the precision. Indeed, on perturbed {\em FBIns} data it improves
the precision from .11 to .98, and improves the F-measure by 403\%.
\eat{ and so improve F-measure.
We also observe an improvement of 1\% on {\em Random}
and 5\% on {\em UB} for F-measure. 
}

\begin{figure}[t]
\begin{center}
\begin{minipage}{.49\linewidth}
\begin{center}
\includegraphics[width=8cm]{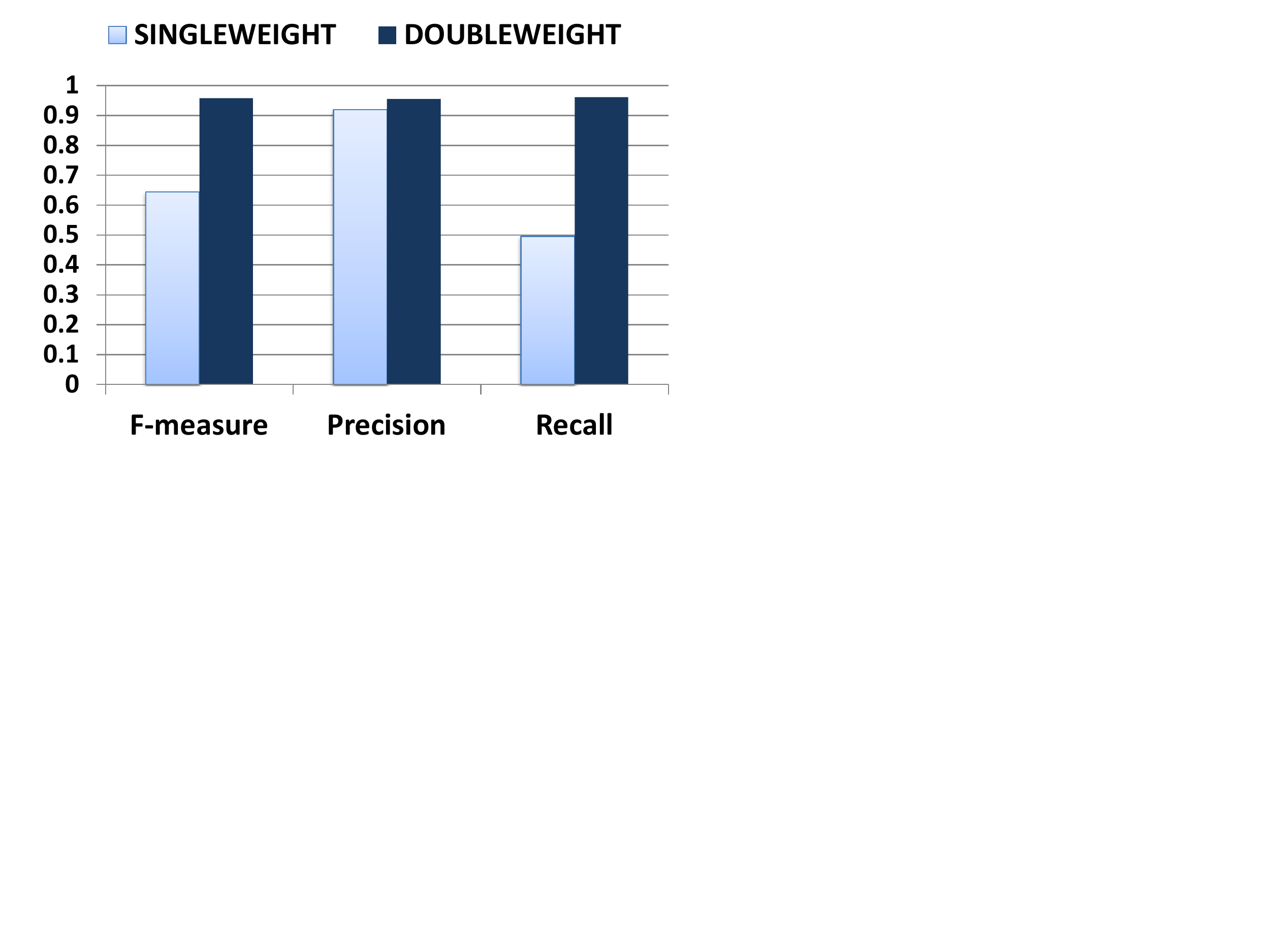}
\vspace{-1.5in}
{\small\caption{Distinct values on \emph{Random} data.\label{fig:agree}}}
\end{center}
\end{minipage}
\hfill
\begin{minipage}{.49\linewidth}
\begin{center}
\includegraphics[width=8cm]{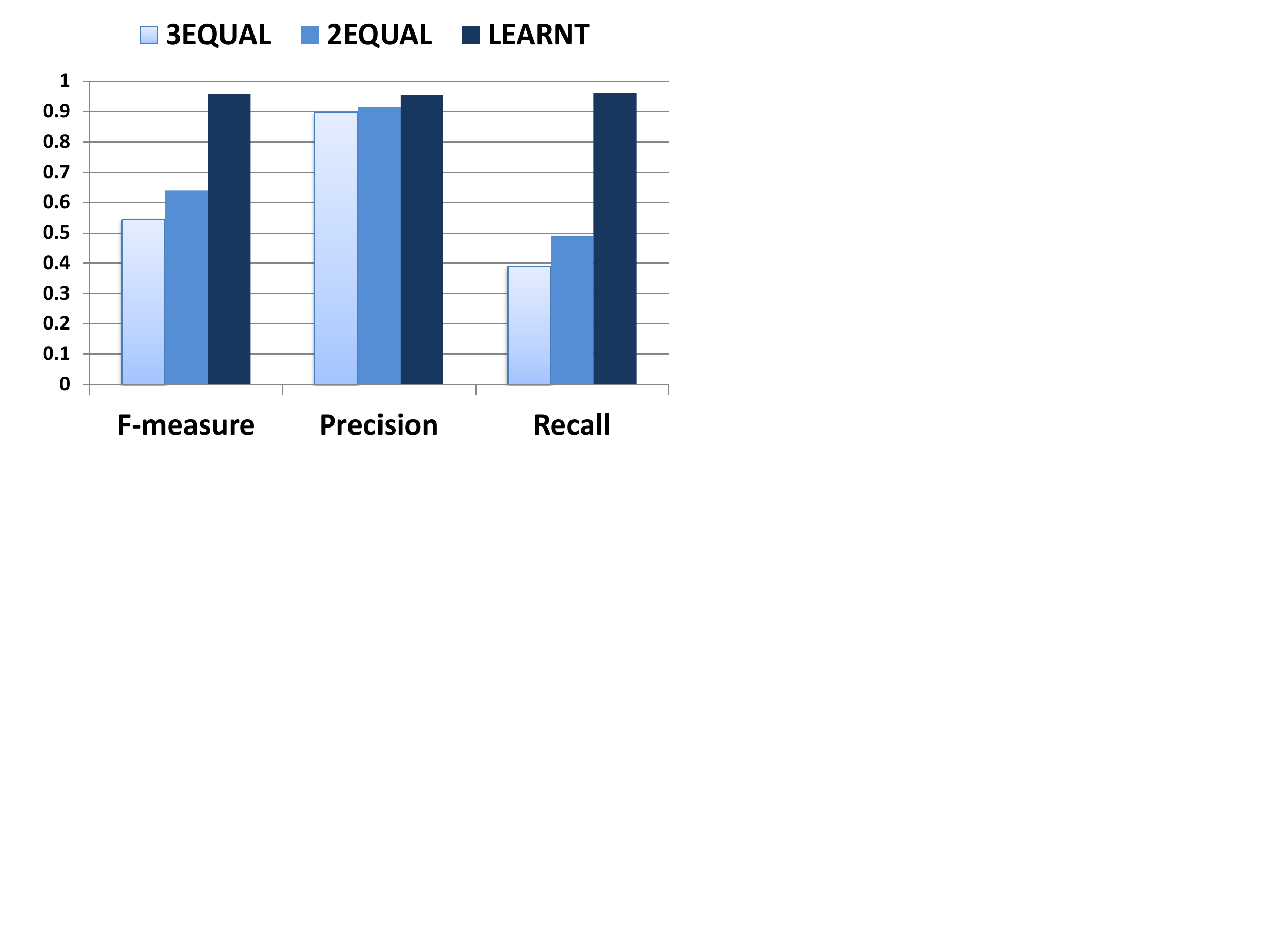}
\vspace{-1.5in}
{\small\caption{Attribute weights on {\em Random} data.\label{fig:equal}}}
\end{center}
\end{minipage}
\end{center}
\vspace{-.2in}
\end{figure}

\smallskip
\noindent{\bf Attribute weight:} We next considered our weight
learning strategy. We first compared {\sc SeparatedDominant},
which learns separated
weights for different dominant-value attributes, and
{\sc UnitedDominant} (our default), which considers all such attributes
as a whole and learns one single weight for them.
Figure~\ref{fig:dmnt} shows that on {\em Random} the latter improves
over the former by 95.4\% on recall and obtains slightly higher
precision, because it penalizes only if neither phone nor URL
is shared and so is more tolerant to different values
for dominant-value attributes. This shows importance of
being tolerant to value variety on dominant-value attributes.

Next, we compared {\sc SingleWeight}, which learns a single weight
for each attribute, and {\sc DoubleWeight} (our default), which learns different
weights for distinct values and non-distinct values for each attribute.
Figure~\ref{fig:agree} shows that
{\sc DoubleWeight} significantly improves the recall
(by 94\% on {\em Random}) since it rewards
sharing of distinct values, and so can link some satellite records
with null values on dominant-value attributes to the chains they
should belong to. This shows importance of distinguishing
distinct and non-distinct values.

We also compared three weight-setting
strategies: (1) {\sc 3Equal} considers common-value attributes,
dominant-value attributes, and multi-value attributes,
and sets the same weight for each of them;
(2) {\sc 2Equal} sets equal weight of .5 for
common-value attributes and dominant-value attributes,
and weight of .1 for each multi-value attribute;
(3) {\sc Learned} applies weights learned from labeled data.
Recall that by default we applied {\sc Learned}.
Figure~\ref{fig:equal} compares their results.
We observe that (1) {\sc 2Equal} obtains higher F-measure than
{\sc 3Equal} (.64 vs. .54), since it distinguishes between strong and weak indicators for record similarity; (2) {\sc Learned} significantly outperforms the other two strategies (by 50\% over {\sc 2Equal} and by 76\% over {\sc 3Equal}), showing effectiveness of weight learning. This shows importance of weight learning.



\begin{figure}[t]
\begin{center}
\includegraphics[width=8cm]{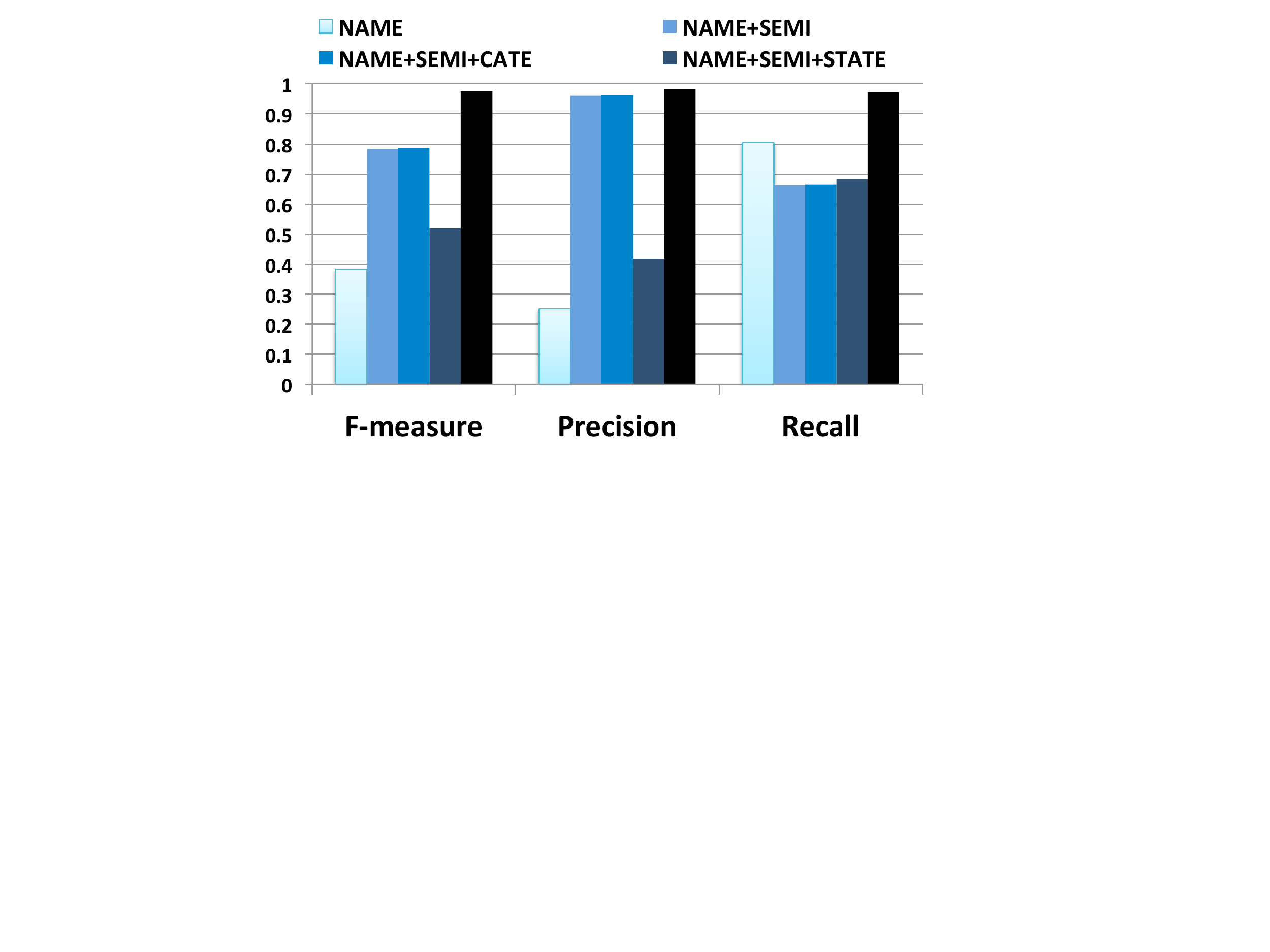}
\vspace{-1.4in}
{\small\caption{Attribute contribution on perturbed {\em FBIns}.\label{fig:attr-cntr}}}
\end{center}
\vspace{-.25in}
\end{figure}

\smallskip
\noindent{\bf Attribute contributions:} We then consider the
contribution of each attribute for chain classification.
Figure~\ref{fig:attr-cntr} shows the results
on the perturbed {\em FBIns} data
and we have four observations.
(1) Considering only {\sf name} but not any other attribute obtains a
high recall but a very low precision, since all listings on this
data set have the same name. (2) Considering dominant-value
attributes in addition to {\sf name} can improve the precision
significantly and improve the F-measure by 104\%.
(3) Considering {\sf category} in addition does not further
improve the results while considering {\sf state} in addition
even drops the precision significantly, since three chains
in this data set contain the same wrong value on {\sf state}.
(4) Considering both {\sf category} and {\sf state} improves
the recall by 46\% and obtains the highest F-measure.

\smallskip
\noindent{\bf Robustness w.r.t. parameters:} We also ran experiments to test robustness against parameter setting. We observed very similar
results when we ranged $p$ from .8 to 1 and $\theta_{th}$ from .5
to .7. 

\subsection{Evaluating efficiency}\label{subsec:effi}
Our algorithm finished in 8.3 hours
on the whole {\em Biz} data set with 18M listings;
this is reasonable given that it is an offline process
and we used a single machine. Note that simpler methods (we describe shortly) 
took over 10 hours even for the first stage on fragments of the {\em Biz} data set. 
Also note that using the Hadoop infrastructure can reduce execution time
for graph construction from 1.9 hours to 37 minutes; 
we skip the details as it is not the focus of the paper.

\eat{
A lot of components of
the algorithm are good candidates for parallelization (\eg, core generation,
clustering), so we can imagine much shorter execution time
in Hadoop infrastructure and would leave this for future work.
}

\smallskip
\noindent
{\bf Stage I:}
It spent 1.9 hours for graph construction and 2.2 minutes
for core generation. To test scalability and understand
importance of our choices for core generation, we randomly divided the whole data set into five subsets of the same size; we started with one subset and gradually added more. 
We compared five core generation methods:
{\sc Naive} applies {\sc Split}
on the original graph; {\sc Index} optimizes {\sc Naive}
by using an inverted index; {\sc SIndex} simplifies
the inverted list by Theorem~\ref{thm:simplify};
{\sc Union} in addition merges v-cliques into v-unions
by Theorem~\ref{thm:union}; {\sc Core} (Algorithm 1)
in addition splits the input graph by
Theorem~\ref{thm:necessary}. Figure~\ref{fig:scalability}(a) shows the results
and we have five observations.
(1) {\sc Naive} was very slow. Even though it applies {\sc Split}
rather than finding the max flow for every pair of nodes,
so already optimizes by Theorem~\ref{thm:maxflow},
it took 6.8 hours on only 20\% data and took more than 10 hours
on 40\% data. (2) {\sc Index} improved {\sc Naive} by
two orders of magnitude just because the index simplifies
finding neighborhood
v-cliques; however, it still took more than 10 hours on 80\% data.
(3) {\sc SIndex} improved {\sc Index} by 41\% on 60\% data
as it reduces the size of the inverted index by 64\%.
(4) {\sc Union} improved {\sc SIndex} by 47\% on
60\% data; however, it also took more than 10 hours on 80\% data.
(5) {\sc Core} improved {\sc Union} significantly; it
finished in 2.2 minutes on the whole data set so further reduced execution
time by at least three orders of magnitude, showing importance of splitting.
Finally, for graph construction, Figure~\ref{fig:scalability}(b) shows 
the linear growth of the execution time.

\smallskip
\noindent
{\bf Stage II:}
After core identification we have .7M cores and 17.3M
satellites. It spent 6.4 hours for clustering:
1.7 hours for blocking and 4.7 hours for clustering.
The long time for clustering is because of the huge number
of blocks. There are 1.4M blocks with multiple elements
(a core is counted as one element),
with a maximum size of 22.5K and an average of 4.2.
On only 35 blocks clustering took more than 1 minute
and the maximum is 2.5 minutes, but for 99.6\% blocks
the size is less than 100 and {\sc Cluster} took less
than 60 ms. The average time spent on each block
is only 9.6 ms.

\begin{figure}[t]
\vspace{-.1in}
\centering
\includegraphics[scale=.35]{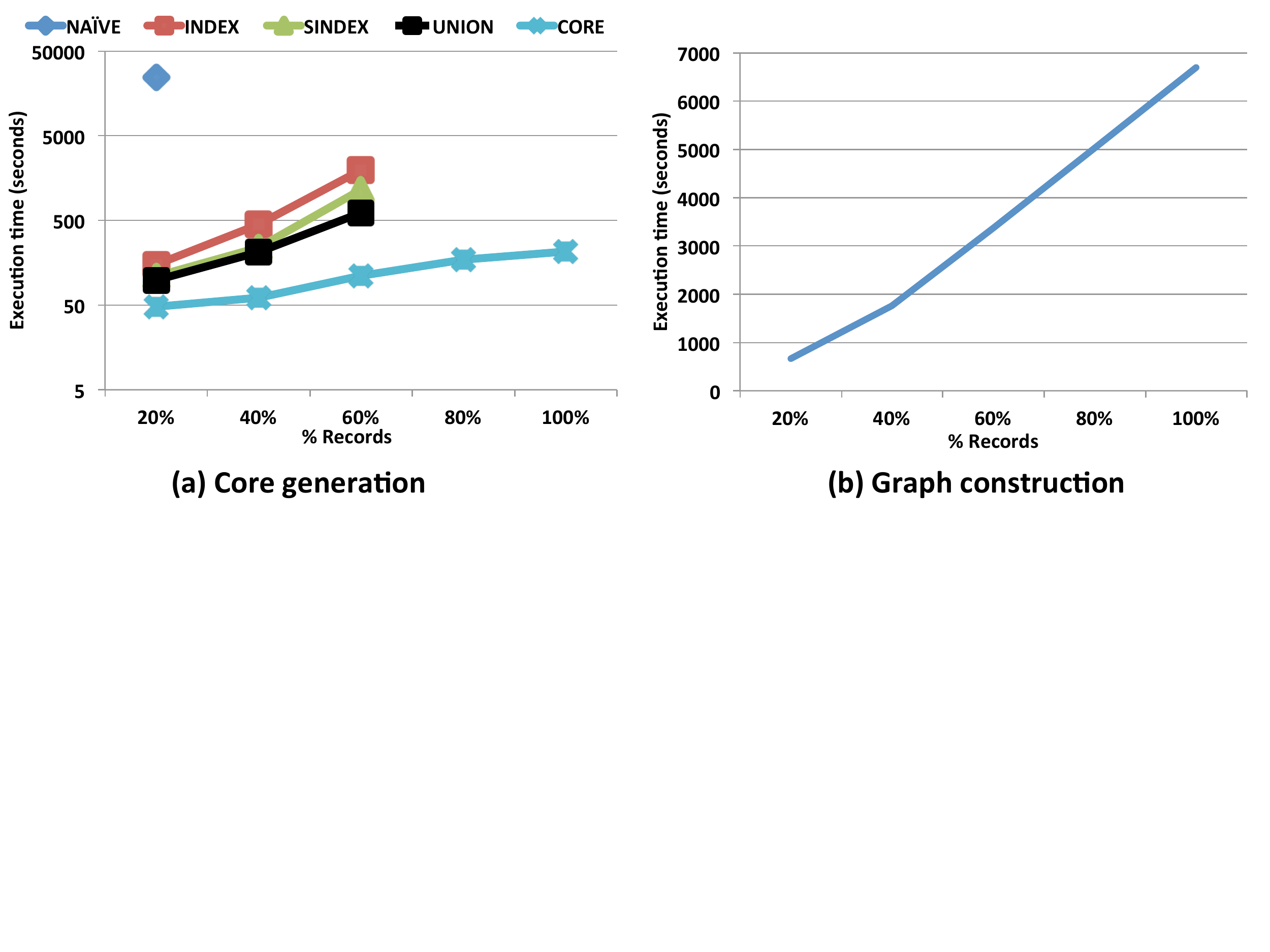}
\vspace{-1.5in}
{\small\caption{Execution time (we plot only those below 10 hours).\label{fig:scalability}}}
\vspace{-.2in}
\end{figure}

\eat{
\begin{figure}[t]
\begin{center}
\begin{minipage}{.48\linewidth}
\begin{center}
\includegraphics[width=7cm]{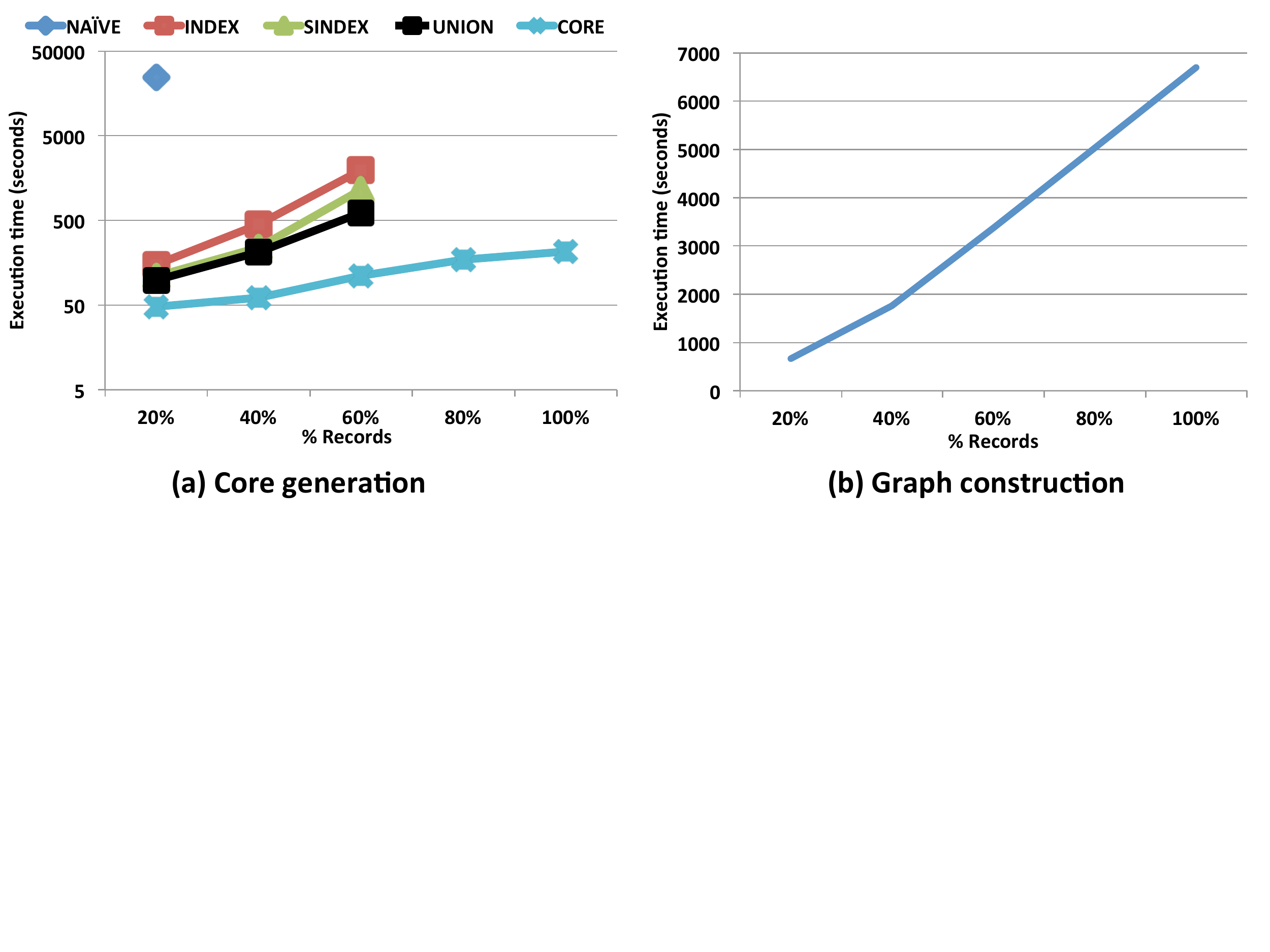}
\vspace{-1.4in}
{\small\caption{Algo. scalability.\label{fig:scalability}}}
\end{center}
\end{minipage}
\hfill
\begin{minipage}{.48\linewidth}
\begin{center}
\includegraphics[width=7cm]{cluster_dist}
\vspace{-1.4in}
{\small\caption{Exe. time distr. of {\sc CLUSTER}.\label{fig:cluster-dist}}}
\end{center}
\end{minipage}
\end{center}
\vspace{-.3in}
\end{figure}
}

\subsection{Summary and discussions}\label{sec:sum}
\smallskip
\noindent
{\textbf{Summary}}: We summarize our observations as follows.
\vspace{-.1in}
\begin{enumerate}\tightlist
\item Identifying cores and leveraging evidence learned from
the cores is crucial in group linkage.
\item There are often erroneous values in real data
and it is important to be robust against them;
applying {\sc OneDom} and requiring $k \in [1, 5]$ already performs
well on most data sets that have reasonable number of errors.
\item Distinguishing the weights for distinct and non-distinct
values, and setting weights of values according to their popularity
are critical for obtaining good clustering results.
\item Our algorithm is robust on reasonable
parameter settings.
\item Our algorithm is efficient and scalable.
\end{enumerate}

\smallskip
\noindent
{\textbf{Discussion}}: In the paper, we present single-machine algorithms to identify groups. Performing such date-intensive tasks on powerful distributed hardwares and service infrastructures has become popular, in particular with the emerging of widely advisable MapReduce programming model~\cite{Shvachko:2010:HDF:1913798.1914427, Chambers:2010:FEE:1806596.1806638, FerreiraCordeiro:2011:CVL:2020408.2020516}. We next discuss possible parallalized solutions of our algorithms in Hadoop infrastructure. 

For graph construction, we can proceed in two steps: (1) to create all cliques where nodes sharing the same common-value and a particular dominant-valued attribute are in the same clique, and (2) to find all maximal cliques. In step (1), we first distribute records and map a record $r$ to one or more $<key, value>$ pairs where $key$ is a value on a particular dominant-value attribute of $r$ and $value$ is the value for common-value attribute of $r$ (Mapper). We then find cliques in each block with a particular $key$, and meanwhile keep an inverted list for each block (Reducer). Step (2) takes the output inverted lists and cliques in Step (1) as input. It first uses each entry in the inverted lists as a $<key, value>$ pair to map cliques, so that all cliques that a record $r$ belongs to are mapped into the same block (Mapper). We then find all maximal cliques within each block (Reducer).

To detect cores in the similarity graphs, the algorithm proceeds iteratively. We can use Spark~\cite{Zaharia:2010:SCC:1863103.1863113}, a cluster computing framework to support iterative jobs while retaining the scalability and fault tolerance of MapReduce. For each iteration, we first partition the input graphs into blocks so that each block contains all records of the same maximal connected component (Mapper), and proceed {\sc Core} within each block in parallel (Reducer). Note that the MapReduce solution may not denominate our single-machine solution that takes only 2.2 minutes, because of the additional overhead of the MapReduce program. 

In similar ways, we identify groups as follows. We first partition the input elements (satellites and cores) into blocks so that each block contains elements that may potentially belong to the same group (Mapper), and proceed {\sc Cluster} within each block in parallel (Reducer).

\section{Conclusions}\label{sec:conclude}
In this paper we studied how to link records to identify groups. 
We proposed a two-stage algorithm that is shown to be empirically 
scalable and accurate over two real-world data sets. 
Future work includes studying the best way to combine record linkage
and group linkage, extending our techniques for finding overlapping 
groups, 
and applying the two-stage framework in other contexts where
tolerance to value diversity is critical.

\bibliographystyle{abbrv}
\bibliography{base}

%
\end{document}